\documentclass[11pt]{article}
\usepackage[letterpaper, margin=1in]{geometry}
\usepackage[utf8]{inputenc}
\def\PrintMode{0}%

\usepackage{tikz}
\usetikzlibrary{positioning}
\usepackage{pgffor}
\usetikzlibrary{decorations.pathreplacing}
\usepackage{expl3}
\usepackage[T1]{fontenc}
\usepackage{graphicx}
\usepackage{float}
\usepackage[lined,ruled,linesnumbered]{algorithm2e}
\usepackage{xspace}
\usepackage[utf8]{inputenc}
\usepackage[noadjust]{cite}
\usepackage{amsmath, amssymb, amsfonts, amsthm}
\usepackage{enumerate}
\usepackage{xcolor}
\usepackage[all]{xy}
\usepackage[linktocpage=true,pagebackref=true,colorlinks,linkcolor=magenta,citecolor=blue,bookmarks,bookmarksopen,bookmarksnumbered]{hyperref}
\usepackage{bm}
\usepackage{bbm}
\usepackage[normalem]{ulem}
\usepackage{paralist}
\usepackage{gensymb}
\usepackage{thmtools}
\usepackage{rotating}
\usepackage{mdframed}
\usepackage{multirow}
\usepackage{appendix}
\usepackage{tabularx}
\usepackage{verbatim}
\usepackage{mathrsfs}
\usepackage{indentfirst}
\usepackage{mleftright}
\usepackage[nottoc]{tocbibind}%
\usepackage{complexity}
\usepackage{tablefootnote}
\usepackage{epigraph}
\usepackage{autobreak}
\usepackage{todonotes}

\ifnum\PrintMode=1

\usepackage{savetrees}
\else
\fi

\graphicspath{{./figs/}}

\usepackage[capitalize]{cleveref}%
\newtheorem{theorem}{Theorem}[section]
\newtheorem{lemma}[theorem]{Lemma}
\newtheorem{proposition}[theorem]{Proposition}
\newtheorem{corollary}[theorem]{Corollary}
\newtheorem{claim}[theorem]{Claim}

\newtheorem{fact}[theorem]{Fact}
\newtheorem{question}[theorem]{Question}

\theoremstyle{definition}
\newtheorem{definition}[theorem]{Definition}

\newtheorem{remark}[theorem]{Remark}
\newtheorem*{remark*}{Remark}

\renewcommand*\backref[1]{\ifx#1\relax \else (cit.~on p.~#1) \fi} %

\makeatletter
\def\moverlay{\mathpalette\mov@rlay}
\def\mov@rlay#1#2{\leavevmode\vtop{%
		\baselineskip\z@skip \lineskiplimit-\maxdimen
		\ialign{\hfil$\m@th#1##$\hfil\cr#2\crcr}}}
\newcommand{\charfusion}[3][\mathord]{
	#1{\ifx#1\mathop\vphantom{#2}\fi
		\mathpalette\mov@rlay{#2\cr#3}
	}
	\ifx#1\mathop\expandafter\displaylimits\fi}
\makeatother

\renewcommand{\poly}{\mathrm{poly}}

\renewcommand{\polylog}{\mathrm{polylog}}

\renewcommand{\emptyset}{\varnothing}

\newcommand{\Enc}{\mathsf{Enc}}

\newcommand{\rank}

\newlang{\MCSP}{MCSP}
\newlang{\MFSP}{MFSP}
\newlang{\MKtP}{MKtP}
\newlang{\MKTP}{MKTP}
\newlang{\itrMCSP}{itrMCSP}
\newlang{\itrMKTP}{itrMKTP}
\newlang{\itrMINKT}{itrMINKT}
\newlang{\MINKT}{MINKT}
\newlang{\MINK}{MINK}
\newlang{\MINcKT}{MINcKT}
\newlang{\CMD}{CMD}
\newlang{\DCMD}{DCMD}
\newlang{\CGL}{CGL}
\newlang{\PARITY}{PARITY}
\newlang{\Empty}{\textsc{Empty}}
\newlang{\Avoid}{\textsc{Avoid}}
\newlang{\Sparsification}{\textsc{Sparsification}}
\newlang{\HamEst}{\mathsf{HammingEst}}
\newlang{\HamHit}{\mathsf{HammingHit}}
\newlang{\CktEval}{\textsc{Circuit-Eval}}
\newlang{\Hard}{\textsc{Hard}}
\newlang{\cHard}{\textsc{cHard}}
\newlang{\CAPP}{CAPP}
\newlang{\GapUNSAT}{GapUNSAT}
\newlang{\OV}{OV}
\renewlang{\PCP}{PCP}
\newlang{\PCPP}{PCPP}
\newclass{\Avg}{Avg}
\newclass{\ZPEXP}{ZPEXP}
\newclass{\DLOGTIME}{DLOGTIME}
\newclass{\ALOGTIME}{ALOGTIME}
\newclass{\ATIME}{ATIME}%
\newclass{\SZKA}{SZKA}
\newclass{\Laconic}{Laconic\text{-}}
\newclass{\APEPP}{APEPP}
\newclass{\SAPEPP}{SAPEPP}

\newclass{\TFSigma}{TF\Sigma}
\newclass{\NTIMEGUESS}{NTIMEGUESS}

\newlang{\Formula}{Formula}
\newlang{\THR}{THR}
\newlang{\EMAJ}{EMAJ}
\newlang{\MAJ}{MAJ}
\newlang{\SYM}{SYM}
\newlang{\DOR}{DOR}
\newlang{\ETHR}{ETHR}
\newlang{\Midbit}{Midbit}
\newlang{\LCS}{LCS}
\newlang{\TAUT}{TAUT}
\newlang{\Poly}{\text{-}Poly}

\renewcommand{\K}{\mathrm{K}}

\newcommand{\rout}{R_{\rm out}}

\newcommand{\F}{\mathbb{F}}
\renewcommand{\K}{\mathbb{K}}

\def\out{\mathsf{out}}

\usepackage{adjustbox}

\renewcommand{\epsilon}{\varepsilon}
\newcommand{\eps}{\epsilon}

\usepackage{diagbox}

\usepackage[draft, inline,marginclue,index]{fixme}  %
\fxsetup{theme=color,mode=multiuser}

\definecolor{color1}{RGB}{46,134,193}

\definecolor{color7}{RGB}{128,0,128}
\FXRegisterAuthor{yeyuan}{ayeyuan}{\color{red}Yeyuan}

\definecolor{color3}{RGB}{255,128,0}
\FXRegisterAuthor{zihan}{azihan}{\color{blue}Zihan}

\definecolor{color5}{RGB}{128,128,128}
\FXRegisterAuthor{task}{atask}{\colorbox{color5}{\color{white}Task}}

\newcommand{\wt}{\operatorname{wt}}

\newif\ifmynotes
\mynotestrue
\title{Bounds and Limitations on Codes Achieving\\List Recovery Capacity\footnote{An extended abstract of this article will appear in the proceedings of SODA 2027. This full version includes complete proofs of all claims.}}
\date{}
\author{Joshua Brakensiek\thanks{University of California, Berkeley. Supported in part by a Simons Investigator award of Venkatesan Guruswami, and NSF awards CCF-2211972 and DMS-2503280 \href{mailto:josh.brakensiek@berkeley.edu}{\texttt{josh.brakensiek@berkeley.edu}}} 
\and
Yeyuan Chen\thanks{Department of EECS, University of Michigan, Ann Arbor. Supported in part by National Science Foundation grant No. CCF-2236931. 
 \href{mailto:yeyuanch@umich.edu}{\texttt{yeyuanch@umich.edu}}}  
\and
Aaron Putterman\thanks{School of Engineering and Applied Sciences, Harvard University, Cambridge, Massachusetts, USA. Supported in part by Simons Investigator Awards to Madhu Sudan and Salil Vadhan, and AFOSR award FA9550-25-1-0112. \href{mailto:aputterman@g.harvard.edu}{\texttt{aputterman@g.harvard.edu}}}
\and
Zihan Zhang\thanks{Simons Institute for the Theory of Computing, UC Berkeley and Institute for Advanced Study, Princeton. \href{mailto:zzhsdj@foxmail.com}{\texttt{zzhsdj@foxmail.com}}}} 

\begin{document}

\maketitle

\begin{abstract}
In coding theory, list recoverability is a fundamental concept which robustly captures how ``spread-out'' codewords are in a code. More formally, given a code $C \subseteq \Sigma^n$ and input lists $S_1, \hdots, S_n \subseteq \Sigma$ of size at most $\ell$, list recoverability requires that there are at most $L$ codewords $c \in C$ such that $c_i \in S_i$ for at least $(1-\rho)n$ choices of $i \in [n]$. List recovery is an important question which has found applications in many areas, including complexity theory, property testing, compressed sensing, streaming algorithms, and cryptography. Despite its widespread influence, basic, fundamental questions in the study of list recoverability remain open including (1) determining the optimal tradeoff between the error radius $\rho$, the size of the recovered list $L$, and the rate of the code $C$ and (2) constructing \emph{explicit} codes achieving (or approaching) such tradeoffs.

As our first main result, we establish a \emph{tight} ``generalized singleton bound'', which exactly characterizes the optimal information-theoretic tradeoff between the error radius $\rho$ and the rate $R$ of the code $C$ in terms of its list-recoverability. Formally, we show that for constant $\ell, L,\rho$ and sufficiently large alphabets $\Sigma$, if we define $R^*=\frac{L+1-\ell}{L}-\frac{L+1}{L}\rho$, it is possible for a $(\rho,\ell,L)$ list-recoverable code to have rate $R^*-\epsilon$ but impossible to have rate $R^*+\epsilon$. One direction of our result already \emph{directly generalizes} and improves a weaker impossibility result due to Goldberg, Shangguan, and Tamo (IEEE TIT 2024).

For our second main result, we prove that there is a fundamental shortcoming in existing methods that aim to construct explicit, optimal list-recoverable codes. Indeed, recent work has constructed explicit codes achieving \emph{list-decoding} capacity (along with other related properties) using a framework introduced in the work of Alon--Edmonds--Luby (AEL) (FOCS 1995). We give a meta-analysis of such constructions by presenting an ``AEL framework'' which captures all such recent constructions in the literature. Within this framework, we show that no AEL-based code can break a recently-identified list-recovery barrier for additive and linear codes. As a result, a fundamentally new construction technique is needed to explicitly achieve list recovery capacity.
\end{abstract}

\pagenumbering{gobble}

\pagebreak

\pagenumbering{arabic}

\section{Introduction}
We study the fundamental tradeoffs between the rate and list recovery of error-correcting codes. The notion of \emph{list recovery} originated as a natural generalization of the more established \emph{list decoding} problem, introduced by Elias \cite{Elias} and Wozencraft \cite{wozencraft1958list} in the 1950s. List decoding itself is a
relaxation of unique decoding in which the decoder is allowed to output a short
list of candidate codewords rather than a single codeword. This relaxation became a central object in
algorithmic coding theory after the development of efficient list-decoding
algorithms for Reed--Solomon codes, beginning with Sudan's algorithm \cite{sudan1997decoding} and the subsequent
Guruswami--Sudan algorithm \cite{guruswami1998improved}.

List recovery emerged as the natural soft-decision analogue of list decoding.
Instead of receiving a single tentative symbol in each coordinate, the decoder is
given a short list of possible symbols at each coordinate and must find all
codewords consistent with these lists on most coordinates. Although the notion
appeared implicitly earlier, the terminology ``list recovery'' was introduced in
the Guruswami--Indyk line of work on efficiently decodable codes
\cite{gi01}. Formally, a code $C \subseteq \Sigma^n$ is
$(\rho,\ell,L)$ list recoverable if, for every collection of input lists $S_1,\ldots,S_n \subseteq \Sigma,$ with $ |S_i| \le \ell$ for all $i\in[n],$
there are at most $L$ codewords $c \in C$ satisfying
$c_i \in S_i$
on at least \((1-\rho)n\) coordinates. The special case \(\ell=1\) is exactly
list decoding, which is abbreviated as  $(\rho,L)$ list decoding. Thus, list recovery asks for a code to remain combinatorially
pseudorandom not merely against Hamming balls, but against all products of
small coordinate-wise sets.

One reason list recovery became so useful is that it composes well with
concatenation \cite{gi01,gi02}. In a concatenated code, one may first decode the inner code blocks,
obtaining a short list of candidate outer symbols in each coordinate. The
remaining task is precisely to list-recover the outer code from these
coordinate-wise candidate sets. This perspective has been used repeatedly in
the construction of efficiently list-decodable and uniquely decodable codes \cite{hrw}.  %

Beyond coding theory, list recoverable codes have also found many influential applications throughout theoretical computer science. They appear in constructions of randomness
extractors and condensers, group testing, compressed sensing and sparse
recovery, streaming algorithms, quantum complexity theory, and collision-resistant hashing \cite{ta2007lossless,guruswami2009unbalanced,GilbertNgoPoratStrauss13,ngo2012efficiently,haitner2015parallel,bostanci2026separating}.
The common theme in these applications is that an algorithm often obtains
partial or ambiguous local information --- a small set of possibilities for each
coordinate, bucket, or measurement --- and list recovery provides a way to
reconcile this local ambiguity into a small number of global candidates.

With these motivations, in this paper we study two basic, fundamental questions concerning list recovery. As our first question, we seek to understand the \emph{optimal}, \emph{information-theoretic} limits of list recovery:

\begin{question}\label{ques:IT-LR}
What is the optimal list recovery a code can have? More precisely, for fixed parameters $\rho, \ell, L$, what is the largest rate $R = \tfrac{\log_{|\Sigma|} |C|}{n}$ a code can have such that it is $(\rho, \ell, L)$ list recoverable?
\end{question}

When $\ell=1$, \Cref{ques:IT-LR} is already solved by two papers \cite{st20,bgm23}. They showed that the largest rate is $R=1-\frac{L+1}{L}\rho$ for a code to be $(\rho,L)$ list decodable. This is called the generalized singleton bound (GSB) for list decoding. However, the optimal rate (i.e., a tight generalized singleton bound) for list recovery is still an open problem.

Despite a number of partial results (e.g., \cite{GST24,RYZ24,CZ25,LMS25,LS25}) no one has been able to fully answer this question even for the simple parameter setting of $\ell = L = 2$. For instance, the work of \cite{GST24} studies this exact question and only proves the optimal rate must be \emph{smaller} than $R<\frac{L+1-\ell}{L}(1-\rho)$. Our first main result is a \emph{tight} generalized Singleton bound for list recovery in the regime for which $\ell$ and $L$ are constants and the alphabet $\Sigma$ is sufficiently large.

\begin{theorem}[Informal, see \Cref{thm:LR-GSB} and \Cref{cor:LR-GSB-tight}]\label{thm:GSB-intro}
Let $L \ge \ell \ge 1$ be constants, and let $R, \rho \in [0, 1]$ be such that
\begin{align}
    R = \frac{L+1-\ell}{L} - \frac{L+1}{L}\rho\label{eq:LR-GSB-rate-intro}
\end{align}
Then, over sufficiently large alphabets, for any constant $\eps > 0$, sufficiently long codes of rate $R+\eps$ are \emph{not} $(\rho, \ell, L)$ list-recoverable. However, there exists codes of rate $R-\eps$ of arbitrarily-long block length which are $(\rho, \ell, L)$ list-recoverable.
\end{theorem}
\cref{thm:GSB-intro} vastly generalizes \cite[Theorem 5.1]{GST24} which showed that $(\rho, \ell, L)$ list recovery is impossible when $R > \frac{L+1-\ell}{L} - \frac{L+1-\ell}{L}\rho$. Their bound only matches our bound in the known $\ell=1$ (list-decoding) setting~\cite{st20,bgm23}. Unlike \cite[Theorem 5.1]{GST24}, we crucially provide a matching existence argument thereby showing that \cref{eq:LR-GSB-rate-intro} is the \emph{exact} rate threshold, and thus \cref{thm:GSB-intro} can be viewed as the true ``generalized Singleton bound'' for list recovery. See \cref{rem:GST} for a more detailed comparison of the results. 

Given our better understanding of optimal list recovery, our second question is whether we can construct \emph{explicit}\footnote{By explicit, we mean that the code has a deterministic polynomial-time encoding function.} codes which achieve list-recovery capacity.\footnote{In this paper, we say that a code of rate $R$ achieves list-recovery capacity if it is $(1-R-\eps, \ell, L)$ list-recoverable with $L = \poly(\ell, 1/\eps)$. Some works allow a weaker bound on $L$ (such as $\poly(n)$) to achieve capacity. As we soon discuss in \cref{fact:LR-limit}, this notion of list-recovery capacity is impossible to achieve for additive codes.}

\begin{question}\label{ques:LR-explicit}
Can we explicitly construct a code $C \subseteq \Sigma^n$ of near-optimal size which is $(\rho, \ell, L)$ list recoverable? What properties must such a code $C$ have?
\end{question}

In this paper, we prove \cref{thm:AEL-intro} which gives a partial negative answer to \cref{ques:LR-explicit} by ruling out a large class of error-correcting codes, i.e., showing that \emph{no codes} from this family can have optimal list recovery. More precisely, we show that a large family of codes based on the Alon--Edmonds--Luby (AEL) framework~\cite{ael95} cannot hope to achieve list-recovery capacity. Due to the technical nature of describing this ``AEL framework,'' we defer the formal statement of \cref{thm:AEL-intro} to \cref{subsec:AEL-framework}.

Nevertheless, \cref{thm:AEL-intro} informally states that any code fitting within the AEL framework  cannot hope to break a recently-identified list recovery barrier for linear and additive codes~\cite{CZ25,LMS25,LS25}. Given a finite field $\F$, we say that $C \subseteq \F^n$ is a linear code if it is a subspace of $\F$. If $\F$ is an extension of another field $\K$, then we can impose a weaker condition that $C \subseteq \F^n$ is merely a $\K$-linear\footnote{That is, $C$ is closed under addition and multiplication by scalars in $\K$, but not necessarily by scalars in $\F$} subspace. 

Although linear codes have been studied for nearly the entire modern history of coding theory, additive codes have found recent prominence in the study of folded Reed--Solomon and other ``subspace-designable'' codes. The modern capacity-achieving story began with Guruswami and Rudra, who
introduced folded Reed--Solomon codes and showed that they can be efficiently
list decoded up to radius \(1-R-\varepsilon\), giving the first explicit
capacity-achieving list-decodable codes over large alphabets \cite{gr08}. Subsequent linear-algebraic decoding methods, including the work of
Guruswami--Wang \cite{GW11} and Kopparty \cite{kopparty2015list}, extended this picture to related polynomial
codes such as derivative and multiplicity codes. These
algorithms also revealed a recurring feature of additive algebraic codes:
the candidate codewords are first captured inside a low-dimensional affine
subspace. 

Another important development for additive codes was the subspace design viewpoint. Guruswami and Kopparty \cite{guruswami2016explicit} gave explicit constructions of subspace designs and
showed their close relationship to folded Reed--Solomon and multiplicity
codes. In modern language, this says that these additive
polynomial codes satisfy strong higher-dimensional analogues of distance.
This perspective has become a unifying way to explain why folded
Reed--Solomon and multiplicity codes have strong list decoding and
list recovery behavior \cite{gr08,krsw,tamo24,srivastava2025improved,CZ25,ashvinkumar2026algorithmic}. The work of Chen and Zhang \cite{CZ25} showed that all these subspace design codes can achieve list decoding capacity with optimal list size. Most recently, the work of Brakensiek, Chen, Dhar, and Zhang \cite{brakensiek2026random} further developed this subspace design viewpoint by showing how such additive codes in fact inherit many ``local properties'' (beyond just list-decodability!) of random linear codes. At the same time, a concurrent work of Jeronimo and Shagrithaya \cite{js26} showed a similar phenomenon for codes based on the famous Alon--Edmonds--Luby (AEL) framework \cite{ael95}. Using this as a starting point, Goyal, Guruswami, and Hsieh \cite{ggh26} were then able to construct explicit constant-alphabet subspace designs via this AEL framework.

Although linear and additive codes have many excellent properties, such as achieving list-decoding capacity \cite{krsw,tamo24,srivastava2025improved,CZ25}, they have recently been proved to be \emph{unable} to achieve list recovery capacity:

\begin{fact}[List recover barrier for additive codes~\cite{CZ25,LMS25,LS25}]\label{fact:LR-limit}
For any $\eps > 0$, any additive code $C \subseteq \F^n$ of rate $R$ cannot be $(1-R-\eps, \ell, \ell^{o(R/\eps)})$ list recoverable.
\end{fact}

A priori, one may have hoped that even though linear and additive codes \emph{cannot} achieve list recovery capacity, that perhaps \emph{non-additive} codes constructed via the AEL framework might still be able to achieve list recovery capacity. The main insight of \cref{thm:AEL-intro} is that this is \emph{not} possible; i.e., even \emph{non-additive} codes within the AEL framework are unable to break the barrier identified in \cref{fact:LR-limit}. As such, if one hopes to find explicit codes matching or approaching \cref{thm:GSB-intro}, one needs to introduce fundamentally new ideas into the construction of such codes. To formally build the AEL framework used in \cref{thm:AEL-intro}, we begin by presenting AEL codes (i.e., codes constructed via this framework) and their applications in the literature.

\subsection{Alon--Edmonds--Luby (AEL) Codes}

In a landmark work, \cite{ael95} constructed explicit near-MDS codes, which are codes with rate $R$, relative distance $1-R-\eps$, over alphabet with constant size $\exp(\polylog(1/\eps)/\eps^4)$. This was the first explicit  construction of near-MDS codes with constant alphabet size and near-linear time algorithms for encoding and unique-decoding \cite{gi02}. Although the alphabet size is larger than the $O(1/\eps^2)$ alphabet size of algebraic-geometry (AG) codes~\cite{tvz82}, the near-linear time encoding and unique-decoding algorithms for AG codes are still open problems.

More importantly, AEL codes give a general framework for constructing codes with desirable properties via a ``local-to-global'' approach. At a high level, if one would like to construct a code with some combinatorial property $\mathcal{P}$ (e.g., $\mathcal{P}$ might refer to $(1-R-\eps,\ell,\ell^{o(R/\eps)})$ list recoverablility in our case) and block length $n$, one can utilize AEL codes as follows: 
\begin{itemize}
\item First, one must prove there \emph{exist} codes with the property $\mathcal{P}$.
\item Next, one designs another code property $\mathcal{P'}$, so that there are explicit ``outer'' codes with rate $R_{\out}=1-\eps$, block length $n$, and property $\mathcal{P'}$. Usually $\mathcal{P}'$ is some weaker property than $\mathcal{P}$. For example, in many previous works \cite{jmst25,st25}, when $\mathcal{P}$ is a desired optimal list-decoding/list-recovery quality, $\mathcal{P}'$ only requires that the outer code has constant relative distance $\Omega_{\eps}(1)$ (which is much weaker than list decoding / recoverability).
\item In addition to the ``outer code'', one must also construct an ``inner'' code, which now does satisfy the target property $\mathcal{P}$. More formally, one constructs inner codes with rate $R_{\mathsf{in}}$, target property $\mathcal{P}$, and block length which we denote by $D$. Of key importance is the fact that $D$ is a function \emph{only} of $\eps$ (rather than $n$), so this inner code can typically be constructed efficiently via exhaustive search. In this way, one does not need to construct asymptotically \emph{explicit} constructions of codes that satisfy the property $\mathcal{P}$, and rather their existence (as mentioned in the first point) suffices.  %
\item Next, one must find a way to ``compose'' these two codes. To define this composition, one requires
a $D$-regular bipartite spectral expander $G=([n]\cup [n],E)$  (which are known to exist from the work of, e.g.,  \cite{lps88}) with small enough normalized second eigenvalue $\lambda_2/D$, where $\lambda_2$ denote the second largest eigenvalue of the adjacency matrix of $G$. 
\item The final step is to concatenate the designed inner and outer codes \cite{concat}. For an outer code $C_{\out}\subseteq\Sigma_{\out}^n$, an inner code $C_{\mathsf{in}}\subseteq \Sigma^D_{\mathsf{in}}$ with its encoder $\Enc\colon \Sigma_{\out}\to C_{\mathsf{in}}$, where $|C_{\mathsf{in}}|=|\Sigma_{\out}|$, their concatenation is a code constructed by encoding each symbol of outer codewords using $\Enc$ as an inner codeword, and concatenating these inner codewords according to the original sequence. This is then followed by using the spectral expander $G$ to ``reshuffle'' and ``fold'' the symbols. This final step is critically important; it leads to a proof that the property $\mathcal{P}$ from the inner code is still \emph{approximately preserved} as one translates from block length $D$ (inner code) to the amplified block length $n$. This heavily relies on the expansion properties of $G$ which yield a sufficiently ``pseudorandom'' reshuffling, provided the outer code satisfies the property $\mathcal{P}'$.

\end{itemize} 

There have been numerous follow-up works which use this AEL framework to construct error-correcting codes with various desired combinatorial properties. For instance, the work of \cite{gr08} constructs $(1-R-\eps,\poly(n))$ list-decodable AEL codes  with alphabet size $\exp(\log{(1/\eps)}/\eps^4)$. In this setting, they use $(\eps,O(1/\eps),\poly(n))$ list-recoverability as the required outer code property $\mathcal{P}'$. In another direction, \cite{hw15} constructs $(1-R-\eps,\ell,O_{\ell,\eps}(1))$ \emph{erasure}-list-recoverable codes\footnote{This is a weaker variant of list recovery, see \cite[Definition 1]{hw15}}. They use $(\Omega(\eps^3),O_{\ell,\eps} (1),O_{\ell,\eps}(1))$ erasure-list-recoverability and $\Omega(\eps^2)$-relative distance as the required outer code property $\mathcal{P}'$. 

This framework has also found applications in the design of \emph{locally} testable / correctable codes: for instance, \cite{highrate} constructs locally correctable codes from $\frac{1-R-\eps}{2}$ errors and near-MDS locally testable codes with low query complexity. They use high-rate LCCs and LTCs as the required outer codes.  Similarly, \cite{gko,hrw,krsw} construct $(1-R-\eps,\ell,O_{\ell,\eps}(1))$ locally list-recoverable AEL codes with constant alphabet size. They use $(\Omega_{\eps,\ell}(1),\Omega_{\ell,\eps}(1),O_{\ell,\eps}(1))$ list-recoverability as the required outer code property $\mathcal{P}'$. 

The AEL framework has also been used in the design of traditional list recoverable and list-decodable codes. For instance, the work of \cite{jmst25} constructs $(1-R-\eps,O(1/\eps))$ list-decodable AEL codes. \cite{st25} constructs $(1-R-\eps,\ell,O_{\ell,\eps}(1))$ list-recoverable AEL codes. More recently, \cite{js26} constructs nearly-optimal \emph{additive} AEL codes in terms of ``local coordinate-wise linear'' properties \cite{LMS25}. \cite{ggh26} constructs nearly-optimal \emph{subspace design} AEL codes. Note that in \emph{all} of these settings, the results achieve constant alphabet size, and use $\Omega_{\ell,\eps}(1)$-relative distance as the required \emph{outer} code property $\mathcal{P}'$.

It is worth mentioning that AEL codes also yield \emph{algorithmic} advantages for various tasks. In this paper and the above list however, we only focus on \emph{combinatorial} properties guaranteed by AEL codes.
\paragraph{List-recoverability of AEL codes.} Since the outer and inner codes used in the AEL framework can both be non-additive code families, the AEL framework with non-additive implementations has become one of the most promising approaches to attack the aforementioned list recovery barrier for additive codes.

Currently, the best list-recoverable codes via the AEL framework \cite{js26,ggh26} only achieve $(1-R-\eps,\ell,(\ell/(R+\eps))^{O(R/\eps + 1)})$ quality. This output list size \emph{does not} outperform the list-recoverability of explicit additive codes and random linear codes \cite{bcdz26b}. It remains elusive whether we could take advantage of non-additivity within the AEL framework and achieve better list-recoverable codes beyond the limit of additive codes.
In a recent survey \cite[Problem 6.1]{rv26}, the authors explicitly propose constructing better list-recoverable codes via AEL codes as an open problem, albeit with optimizing the alphabet size of the code as the primary objective.

In this paper, we give a strong negative answer to this direction, showing that the AEL framework \emph{does not} yield list-recoverability beyond the aforementioned barrier.

\subsection{The AEL Framework}\label{subsec:AEL-framework}

Our main result concerning the construction of AEL codes is a \emph{lower bound} on the output list size assuming the AEL code is ``reasonable'' in a sense we define shortly.

\begin{theorem}[Informal, see \Cref{cor:main}]\label{thm:AEL-intro}
Let $\ell\ge 2,\eps>0$ be constants. Then, there exists a constant $\lambda>0$ such that for any invocation of an inner-encoder insensitive AEL framework with a $\leq \lambda$-spectral expander, if (1) the outer code has rate at least $1-\ell^{-1/\eps}$ or (2) the outer code is additive with rate at least $1-2\eps$, then the AEL framework cannot guarantee to construct $(1-R-\eps,\ell,\ell^{o(R/\eps)})$ list-recoverable codes.
\end{theorem}

The formal definition of the inner-encoder insensitive AEL framework is in \cref{sec:ael}, but we explain in this section some of the details of the framework. In short, the framework used in \Cref{thm:AEL-intro} requires three conditions (1) a lower bound on the outer code rate, (2) a certain ``insensitivity'' condition for the inner-encoder, and finally (3) a suitable spectral bound. It is worth noting that, to the best of our knowledge, all prior work on AEL codes use proof strategies that satisfy these requirements, thus making \cref{thm:AEL-intro} a highly general result. We  give a more detailed justification for the requirements of the AEL framework in \Cref{sec:rate,sec:encoder,sec:spectral}, along with the definition of being ``inner-encoder insensitive.''

We consider \cref{thm:AEL-intro} to be a coding-theoretic analogue of the famous natural proofs barrier concerning \textsf{P} vs \textsf{NP}~\cite{RR97}. In particular, the natural proofs barrier shows that any sufficiently ``combinatorial proof'' cannot prove that $\mathsf{P} \neq \mathsf{NP}$. Here the inner-encoder insensitivity comes from the same motivation: all previous works on list-decoding(recovery) of AEL codes are actually ``combinatorial,'' which means they only require the inner codes and outer codes certain combinatorial problems from which they can deduce that the corresponding AEL codes satisfying the particular list-recovery bound. These combinatorial proofs hold regardless of the choices of inner encoder, taking advantage of a philosophy shared by all aforementioned works that merely combinatorial advantages in each individual component lead to good list-recoverable codes. Our negative result refutes all of these combinatorial proofs: If you want sufficiently good list-recovable codes, you have to rely on a specialized inner encoder, which does not come from combinatorial advantages of the components.

In \Cref{sec:rate,sec:encoder,sec:spectral}, we discuss the validity of these conditions from three aspects: (1) how does previous work satisfy these conditions, (2) why these conditions are expected to hold for any traditional invocation of the AEL framework, and (3) if an alternative AEL framework does not satisfy one of our conditions (which means our negative result does not apply to it), then in which aspects the AEL framework has to be unusual.

In summary, for those who are pessimistic about the list recovery of AEL codes, our result gives a strong negative lower bound. If one is still optimistic that a more generous AEL framework could be used to break the list recovery barrier for additive codes, our negative results gives \emph{requirements} that any version of the AEL framework \emph{must satisfy} to accomplish such a task. Indeed, any such use of the AEL framework to build better list-recoverable codes has to violate one or more conditions listed in \Cref{sec:rate,sec:encoder,sec:spectral}, which implies the code has at least one of the following four properties:
\begin{enumerate}
\item The rate of inner code is significantly larger\footnote{In the formal statement of \cref{thm:AEL-intro}, we specify that the rate of the inner code times the rate of the outer code equals with the target rate. Thus, an equivalent formulation of this point is that the rate of the outer code must be significantly less than $1$.} than the target rate, which is counter-intuitive. (Paragraph 1 in \Cref{sec:rate})
\item The invocation of the AEL framework \emph{only works} when the normalized second eigenvalue $\lambda_2/D$ is \emph{larger} than some constant, which  is also counter-intuitive, as one typically expects better expansion to yield better codes. (\Cref{sec:spectral}) 
\item The constructions of the \emph{outer codes} themselves are intrinsically non-additive. Note that such codes are not currently known, as the outer code must be \emph{explicit}. Moreover, the proof cannot generalize to outer codes whose redundancy is an arbitrarily small constant. That is, there is a constant $0<c<1$ such that the construction fails when the rate of outer code is larger than $c$. (Paragraph 2 in \Cref{sec:rate})
\item The construction is sensitive to the inner-encoder, which indicates the proof must be whitebox in terms of (i.e., explicitly depend on) the inner encoder that is used. (\Cref{sec:encoder}) 

\end{enumerate}
Each of these four properties is non-traditional in the coding theory literature. \cref{thm:AEL-intro} can be interpreted as saying that at least one of these non-traditional assumptions must be used to break the additive codes barrier for list recovery (\cref{fact:LR-limit}). 

\subsubsection{Outer Code Rate Lower Bounds.}\label{sec:rate}

\paragraph{Additive outer code: rate at least $1-2\eps$.} 
We first explain why an outer code rate lower bound of $1-2\eps$ is expected from a typical invocation 
of local-to-global amplification. As in the above discussion, we let $\mathcal{P}$ denote the target property. Usually $\mathcal{P}$ is defined with respect to the final desired code rate $R$, such as $(1-R-\eps)$-relative distance, $(1-R-\eps,L)$-list-decoding, or $(1-R-\eps,\ell,L)$-list recovery. In our case, $\mathcal{P}$ is $(1-R-\eps,\ell,\ell^{o(R/\eps)})$-list-recoverability. By impossibility results on the rate-radius trade-off \cite{singleton,st20,GST24}, including our (tight) generalized Singleton bound \cref{thm:GSB-intro}, the property $\mathcal{P}$ is \emph{only possible} for codes with rate at most $R+\eps$. 

Since the spirit of AEL framework is a local-to-global amplification on block length, the property $\mathcal{P}_{\mathsf{in}}$ required for the inner code should be no weaker than $\mathcal{P}$. Otherwise, since the local-to-global amplification on block length conceptually incurs some loss in parameters, it does not make sense that a \emph{weaker} property $\mathcal{P}_{\mathsf{in}}$ becomes a \emph{stronger} property $\mathcal{P}$ after the amplification on the block length.

Naturally then, a stronger property $\mathcal{P}_{\mathsf{in}}$ conceptually requires smaller rate, so we should consider $R_{\mathsf{in}}\leq R+\eps$. In the context of code concatenation, we let $R_{\mathsf{in}},R_{\out}$ denote the rate of inner and outer code respectively. Then, after concatenation, the resulting code has rate $R=R_{\mathsf{in}}R_{\out}$. This calculation then already forces that $R_{\out}\ge\frac{R}{R_{\mathsf{in}}}\ge 1-\frac{\eps}{(R+\eps)}\ge 1-\frac{2\eps}{R_{\mathsf{in}}}$. This is exactly the rate lower bound required in the second condition of the formal statement \Cref{cor:main} (We assume a stronger condition $1-2\eps$ in the informal statement only for simplicity). Thus, this first condition is extremely pervasive, and automatically holds for all previous work building AEL codes with additive outer codes. Moreover, by the above reasoning, any ``natural instantiation'' of the AEL framework also satisfies this condition that outer code has rate $\geq 1 - 2\eps$, thereby providing a compelling barrier towards using this framework to break the list-recovery barrier.

Indeed, to bypass this obstruction, one requires an additive outer code which has rate \emph{less than} $1-2\eps$. In this case, then the inner code rate $R_{\mathsf{in}}\ge R+\eps$ must be significantly larger than the target rate, which is counter-intuitive from the previous discussion.
\paragraph{Non-additive outer code: rate at least $1-\ell^{-1/\eps}$.} 
The aforementioned discussion focuses on the setting when the outer code is \emph{additive}. Here, we now discuss restrictions when the outer code is \emph{non-additive}.\footnote{As discussed in \cref{subsec:AEL-sketch}, the rate lower bound on additive outer codes is sharper due to similar structures identified in the proof of \cref{fact:LR-limit} that additive codes by themselves cannot achieve list recovery capacity.} Before this, it is worth mentioning that \emph{no} previous work on AEL codes intrinsically requires non-additive codes as outer codes. Concretely, let $\mathcal{P}'$ denote the outer code property required. In all prior work the outer codes with property $\mathcal{P}'$ can be already explicitly constructed as \emph{additive} codes, and so these constructions are already precluded from breaking the list recovery barrier by the aforementioned point. 

Regardless of non-additivity, we claim that all previous state-of-the-art upper bounds on output list sizes of list decodable/recoverable AEL codes \cite{ael95,jmst25,st25,js26,ggh26} still hold when the outer code rate is larger than $1-\ell^{-1/\eps}$. This implies that all instantiations of the AEL framework in these works \emph{still satisfy} condition (1) of \cref{thm:AEL-intro}, and are thereby precluded from breaking the list-recovery barrier.

Indeed, in the original \cite{ael95} and all recent works on AEL codes \cite{jmst25,st25,js26,ggh26}, there is a function $f\colon (0,1)\to (0,1)$ such that for any  $\mu>0$, as long as the normalized second eigenvalue of the expander satisfies $\lambda_2/D\leq f(\mu)$, the proof strategy still works for any outer code with relative distance at least $\mu$. Although a smaller second eigenvalue requires larger degree $D$, which means larger alphabet size, in this problem we only care about the output list size $\ell^{o(R/\eps)}$ regardless of the alphabet size.

\cite{lps88} guarantees optimal  spectral expanders with $\lambda_2\leq 2\sqrt{D-1}$. Therefore, it suffices to choose  $D\ge4/f^2(\mu)$. The only upper bound for $D$ is $D\leq n$, so any constant $\mu>0$ is a valid parameter setting. This indicates that if we do not care about alphabet size, then for   arbitrarily small constant  $\eta>0$, there is a constant 
$0<\mu\leq \eta$ so that there are outer codes with rate at least $1-\eta$ and relative distance $\mu$, and  \cite{ael95,jmst25,st25,js26,ggh26} can be implemented on these outer codes.\footnote{\cite{highrate} Chooses subconstant $\mu=o(1)$.  Our lower bound works for all outer codes with subconstant relative distance $\mu=o(1)$.}\footnote{\cite{gr08,krsw} requires the relative distance $\mu$ to be large enough so that the outer code has good enough list-decodability/recoverability. However, their proof strategies are now completely and significantly improved by \cite{jmst25,st25}, which works for arbitrarily small constant relative distance $\mu>0$.} By choosing any constant $\eta<\ell^{-1/\eps}$, this implies that all proof strategies in these previous results \emph{still work} under the above outer code rate condition. This precludes these works from breaking the list recovery barrier.

So, in order to break this list recovery barrier, one possibility is to use an instantiation of the AEL framework which \emph{does not} satisfy the rate requirement $1-\ell^{-1/\eps}$. What would such a framework look like? The key fact is that its proof strategy must \emph{intrinsically fail} for small enough constant relative distance $\mu<\ell^{-1/\eps}$, which is conceptually different from all aforementioned works. We believe that would require brand new ideas if such an AEL framework was indeed proposed.

Beyond the above, using the AEL framework with non-additive outer codes would also require \emph{explicit non-additive codes} with properties not guaranteed by additive codes! Indeed, to propose a new AEL framework which bypasses this condition, one would require $\mathcal{P}'$ to be some property that is beyond the reach of additive codes. There exist properties that are proven to be impossible for additive codes like aforementioned list recovery barrier for additive codes (\Cref{fact:LR-limit} or \cite[Section 4]{GST24}). But importantly, none of these properties are known to have non-additive explicit constructions. Therefore, as an inherent prerequisite for this kind of AEL framework, one would first have to worry about how to explicitly construct the non-additive outer code with a natural property $\mathcal{P}'$, which genuinely relies on non-additivity. We leave this as an open problem.

It is worth noting that this stronger rate lower bound is necessary for our negative result: If one allows the non-additive outer codes to have rates (say) $1-\Omega(\eps)$, then it is possible to achieve $(1-R-\eps,\ell,\poly(\frac{\ell}{\eps}))$ list-recoverable codes with outer codes which are $(\Omega(1),O(\frac{\ell}{\eps}),\poly(\frac{\ell}{\eps}))$ list-recoverable via subspace design codes plus subspace evasive sets~\cite{pr04,G11,dl12,note13}. However, such objects like subspace evasive sets are combinatorial objects that are strictly harder to construct compared to $(1-R-\eps,\ell,\poly(\frac{\ell}{\eps}))$ list-recoverable codes, so we do not expect to use them as outer codes merely for list-recoverability.

\subsubsection{Inner-encoder insensitivity}\label{sec:encoder}

Let us first define what inner-encoder insensitivity means. The AEL framework can be seen as code concatenation \cite{concat} plus shuffling via an expander. Specifically, let $C_{\out}\subseteq \Sigma_{\out}^n$ and $C_{\mathsf{in}}\subseteq \Sigma_{\mathsf{in}}^D$ denote the outer and inner code where $|C_{\mathsf{in}}|=|\Sigma_{\out}|$, and $\Enc\colon \Sigma_{\out}\to C_{\mathsf{in}}$ denote some encoder of $C_{\mathsf{in}}$. Then, the concatenation of $C_{\out}$ and $C_{\mathsf{in}}$ is defined as
\[
C=C_{\out}\circ C_{\mathsf{in}}=\{(\Enc(c_1),\dots,\Enc(c_n))\colon c\in C_{\out}\}.
\]
We say an AEL framework is \emph{inner-encoder insensitive} if and only if, for any fixed choice of the chosen $C_{\out}$ and $C_{\mathsf{in}}$, the concatenated code $C_{\out}\circ C_{\mathsf{in}}$ has the target property $\mathcal{P}$ after the expander reshuffle \emph{regardless} of the choice of $\Enc$. In other words, it means the success of the AEL framework only depends on the \emph{combinatorial} properties of outer and inner codes, rather than the encoders used to connect them.

All prior work on AEL codes achieving good combinatorial properties are inner-encoder insensitive, thereby ensuring (in connection with the above discussion) that all known AEL proof strategies are precluded from breaking the list recovery barrier.\footnote{\cite{js26} and \cite{ggh26} have to use a linear encoder in the concatenation. However, their goals are to construct good subspace design codes/LCL preserved codes. These properties are only defined for additive codes, so their encoder choices are compulsory by definitions for their purposes.  In addition, they care about code properties that are tailored for additive codes, while our target $(1-R-\eps,\ell,\ell^{o(R/\eps)})$ list recovery is proven to be beyond the performance of any additive codes.} An inner-encoder \emph{sensitive} AEL framework on the other hand, would have to crucially make use of some special properties provided by specifically chosen inner code encoders. This would require intrinsically brand new ideas beyond all existing proof strategies for AEL codes.

\subsubsection{Spectral Bounds}\label{sec:spectral}

\paragraph{Expanders with small enough normalized second eigenvalue are required.} If we do not care about the alphabet size and assume the explicitness of inner codes, then all previous instantiations of the AEL framework work for all \emph{small enough} values of the normalized second eigenvalue $\lambda_2/D$. This fact is intuitive because a smaller second eigenvalue implies \emph{better} pseudorandomness (Reflected as smaller error terms in the mixing lemma \Cref{lem:expander}). Therefore, in principle we should expect that a smaller second eigenvalue yields better parameters in target properties like the output list size of list recovery. Our impossibility result is applicable to \emph{all} 
$(\leq\lambda)$-expanders with small enough second eigenvalue, where $\lambda>0$ is a constant unrelated to $n$. In this way, all prior work produces AEL codes that follow this condition, and therefore are precluded from breaking the list recovery barrier (again, in conjunction with the above).

\section{Technical Overview}

\subsection{Generalized Singleton Bound for List Recovery}

We start by sketching the ideas that underlie \cref{thm:GSB-intro}.

\paragraph{Upper Bounding the Rate}

We begin by discussing the rate upper bound from \cref{thm:GSB-intro}. To help convey the ideas, we consider the simpler case of $\ell = L =2$. Recall that in this case, a code $C \subseteq \Sigma^n$ fails to be $(\rho, 2, 2)$-list recoverable when there are \emph{three} distinct codewords $c_1, c_2, c_3 \in C$ and coordinate-wise lists $S_1, S_2, \dots S_n \subseteq \Sigma$ each of size $2$ such that each of $c_1, c_2, c_3$ is captured by these lists on at least $(1 - \rho)n$ of their coordinates. Our bound shows that this failure is \emph{unavoidable} for every code of rate $R$ once $\rho \geq \frac{1 - 2R}{3}$, or equivalently, once the level of agreement satisfies $1 - \rho \leq \frac{2 + 2R}{3}$.  

To argue that this failure is unavoidable, we show the existence of a simple ``agreement triangle'' pattern. To start, we suppose that $|C| \geq q^{Rn + \Omega(1)}$. We then choose two disjoint coordinate blocks $I_1, I_2 \subseteq [n]$ with $|I_1| \approx |I_2| \approx Rn$, as long as this is possible. Our starting goal is then to find some codeword $c_1$ such that \emph{many} other codewords agree with $c_1$ on both $I_1$ and $I_2$. Proving this relies on the pigeonhole principle: for a coordinate set $I \subseteq [n]$, there are only $q^{|I|}$ many possible restrictions to $I$. Thus, the set of codewords whose $I$-fiber\footnote{For a codeword $c$, its $I$-fiber is $\{c' \in C: c'|_I = c|_I\}$. } is of size $1$ is bounded exactly by $q^{|I|}$. By applying this reasoning to $I_1, I_2$ and using the fact that our code has at least $q^{Rn + \Omega(1)}$ %
many codewords, this means that we can find a codeword $c_1 \in C$ whose $I_1$-fiber and $I_2$-fiber are \emph{both} of size $\geq 2$.

To summarize, this means that we can find codewords $c_1, c_2, c_3$ such that:
\[
c_2 \neq c_1, c_2|_{I_1} = c_1|_{I_1}, \quad \quad \quad c_3 \neq c_1, c_3|_{I_2} = c_1|_{I_2}.
\]
That is, the codewords $c_1, c_2$ agree on all coordinates in $I_1$, and the codewords $c_1, c_3$ agree in all coordinates in $I_2$.

Now, it remains only to show that this agreement pattern necessarily violates $(\rho, 2, 2)$-list recoverability: on a coordinate $i \in I_1$, the three symbols $c_{1, i}, c_{2, i}, c_{3,i}$ only take two values, as $c_1$ and $c_2$ agree. On a coordinate $i \in I_2$, the same is true, as $c_1$ and $c_3$ agree. So, on the covered blocks $I_1 \cup I_2$, the lists $S_i: i \in I_1 \cup I_2$ are all of size $\leq 2$. If $I_1 \cup I_2 = [n]$ (i.e., they cover all coordinates), then this even gives an outright obstruction to zero-error list recovery. Otherwise, we define $I_0 = [n] - (I_1 \cup I_2)$. Note that on these coordinates, we have no guarantee that the codewords $c_1, c_2, c_3$ agree, and so it is possible that $\{c_{1, i}, c_{2, i}, c_{3, i}\}$ has size $3$. This is where we make our final observation: at each leftover coordinate, we are \emph{still allowed} to use a list of size $2$. Thus by cycling over the sets $\{c_{1, i}, c_{2, i}\}, \{c_{1, i}, c_{3, i}\}, \{c_{2, i}, c_{3, i}\}$ we can ensure that each of the three codewords is included in at least a $2/3$ fraction of the remaining sets! Thus, each of $c_1, c_2, c_3$ is captured on at least\footnote{Assume $|I_0|$ is divisible by $3$.} 
\[
|I_1 \cup I_2| + \frac{2}{3} \cdot |I_0|
\]
many coordinates. Plugging in our choices of $|I_1|, |I_2|$, this implies that each of $c_1, c_2, c_3$ agrees with the lists in at least $\left ( 1 - \frac{1-2R}{3}\right ) \cdot n$ many coordinates. We summarize this argument in \cref{fig:GSBupper}.

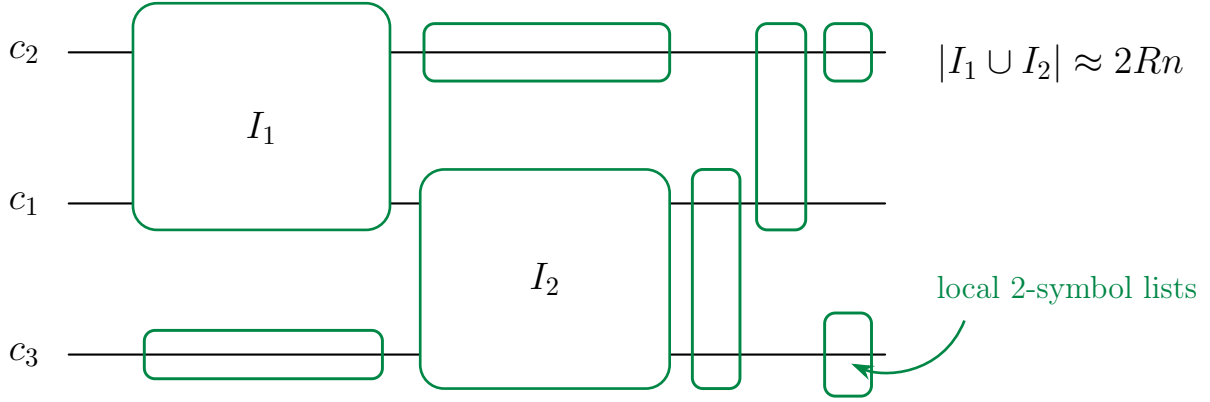
\begin{figure}
    \centering
    \usetikzlibrary{arrows.meta,calc}

\definecolor{listgreen}{RGB}{0,135,70}

\begin{tikzpicture}[
    x=1cm,
    y=1cm,
    codeword/.style={
        draw=black,
        line width=0.8pt,
        line cap=round
    },
    bigblock/.style={
        draw=listgreen,
        fill=white,
        line width=1.0pt,
        rounded corners=9pt
    },
    listbox/.style={
        draw=listgreen,
        line width=1.1pt,
        rounded corners=4pt
    },
    label/.style={
        font=\Large
    },
    mathlabel/.style={
        font=\Large
    },
    greenlabel/.style={
        font=\large,
        text=listgreen
    }
]

\coordinate (c1L) at (0,2);
\coordinate (c1R) at (10.8,2);
\coordinate (c2L) at (0,0);
\coordinate (c2R) at (10.8,0);
\coordinate (c3L) at (0,-2);
\coordinate (c3R) at (10.8,-2);

\draw[codeword] (c1L) -- (c1R);
\draw[codeword] (c2L) -- (c2R);
\draw[codeword] (c3L) -- (c3R);

\node[label,anchor=east] at (-0.25,2) {$c_2$};
\node[label,anchor=east] at (-0.25,0) {$c_1$};
\node[label,anchor=east] at (-0.25,-2) {$c_3$};

\draw[bigblock] (0.85,-0.35) rectangle (4.25,2.65);
\node[mathlabel] at (2.55,1.0) {$I_1$};

\draw[bigblock] (4.65,-2.45) rectangle (7.95,0.45);
\node[mathlabel] at (6.30,-1.0) {$I_2$};

\draw[listbox] (1.00,-2.32) rectangle (4.15,-1.68);

\draw[listbox] (4.70,1.62) rectangle (7.950,2.38);

\draw[listbox] (10.00,1.62) rectangle (10.62,2.38);

\draw[listbox] (8.25,-2.45) rectangle (8.88,0.45);
\draw[listbox] (9.10,-0.35) rectangle (9.75,2.38);
\draw[listbox] (10.00,-2.55) rectangle (10.62,-1.45);

\node[mathlabel,anchor=west] at (11.35,1.85)
    {$|I_1 \cup I_2| \approx 2Rn$};

\node[greenlabel,anchor=west] at (11.35,-1.15)
    {local $2$-symbol lists};

\draw[
    listgreen,
    line width=1.0pt,
    -{Stealth[length=3mm,width=2mm]}
]
    (11.85,-1.55)
    .. controls (11.55,-2.30) and (10.75,-2.30)
    .. (10.35,-2.12);

\end{tikzpicture}
    \caption{A view of the agreement sets $I_1, I_2$ and remaining coverage from coordinate sets. Green sets represent symbols that are covered by the sets $S_1, \dots S_n$. The $I_1$ and $I_2$ boxes represent that the coordinates of $c_1, c_2$ (respectively $c_1, c_3$) exactly coincide.}
    \label{fig:GSBupper}
\end{figure}

The full theorem uses a generalization of this reasoning to capture all choices of $\ell$ and $L$. Rather than using three codewords and lists of size $2$, one can generalize the argument to use $L+1$ codewords and lists of size $\ell$. One then chooses $L$ disjoint coordinate blocks, and for each non-base codeword $c_j \neq c_1$, one forces agreement with the base codeword on a carefully chosen cyclic union $J_j$ of $L +1 - \ell$ many blocks. This design enforces the property that every coordinate belongs to exactly $L +1 - \ell$ of these $J_j$'s. This means that at each such coordinate at least $L+2 - \ell$ of the $L+1$ codewords share the same symbol, implying that all codewords can be covered by a list of size $\ell$. The remaining coordinates (i.e., not in any of $L$ coordinate blocks) are then handled by a similar averaging argument akin to the $\ell = L = 2$ case, which will yield that an $\ell$-sized list can capture $\ell$ out of the $L+1$ many codewords. All in all, this then yields the final agreement bound of 
\[
\frac{\ell}{L+1} \cdot n + \frac{L}{L+1} \cdot Rn,
\]
which is the complement of our generalized Singleton bound:
\[
\rho = \frac{L+1 - \ell}{L+1} - \frac{L}{L+1} \cdot R.
\]

\paragraph{Lower Bounding the Rate.}

To prove a tight characterization of when list recovery is achievable, we need to show that there are codes which can in fact \emph{match} the aforementioned upper bound on the rate. The second part of our proof shows that this is indeed the case. The key quantity we use to show this is possible is the ``collision weight'' of a set of $L+1$ codewords:
\[
\wt(c_1, \dots c_{L+1}) = \sum_{i = 1}^{n} (L+1 - |\{c_{1,i}, \dots c_{L+1, i}\}|).
\]
For our exposition, we will focus on the case when $\ell = L = 2$, in which case $\wt(c_1, c_2, c_{3}) = \sum_{i = 1}^{n} (3 - |\{c_{1,i}, c_{2, i}, c_{3, i}\}|)$.
For intuition, a coordinate where all three codewords are distinct contributes collision weight $0$, and a coordinate where all three codewords are the same contributes weight $2$. 

The starting observation is that for any sequence of two symbol lists $S_1, \dots S_n$ the total number of captured incidences satisfies 
\[
\sum_{j = 1}^3 |\{i \in [n]: c_{j, i} \in S_i\}| \leq 2n + \wt(c_1, c_2, c_3).
\]
The $2n$ follows from the trivial number of incidences that can be captured from two-symbol lists while $\wt(c_1, c_2, c_3)$ captures the ``extra incidences'' that can be captured when symbols from the codewords coincide. Note that if three codewords $c_1, c_2, c_3$ are all captured on at least a $(1 - \rho)n$ fraction of their coordinates, then $3(1 - \rho)n \leq 2n + \wt(c_1, c_2, c_3)$, which means that $\wt(c_1, c_2, c_3) \geq (1-3\rho)n$.

By our  previous discussion, our target Singleton bound in this regime is at $\rho = ( 1- 2R) / 3$; equivalently, this means that $\wt(c_1, c_2, c_3) \geq 2Rn$. All this is to say that avoiding list recovery violations is essentially \emph{equivalent} to showing that every triple of codewords has collision weight at most $2Rn$.

All that remains then is to show that there are codes which satisfy this condition. For this, we rely on a randomly chosen code along with a secondary ``clean-up phase'' of a random code which does not throw out too many codewords (compare  \cite[Proposition I.4]{AGL25}). Indeed, let us start by taking a random code with say $2 q^k$ many codewords for $k = Rn$. We say that a triple of codewords $c_1, c_2, c_3$ is \emph{bad} if $\wt(c_1, c_2, c_3) \geq t_2$, where intuitively, $t_2 \approx 2Rn$, but we will build in some slack which becomes arbitrarily small over large alphabets; more formally $t_2 = 2k + O(n / \log q)$.

Now, we analyze these three codewords $c_1, c_2, c_3 \in \Sigma^n$: at each coordinate $i$, we build a partition of $\{1, 2, 3\}$ which captures the equality patterns of these codewords. That is, a partition $\{1, 2\}, \{3\}$ reflects that $c_{1, i} = c_{2, i}$ at this coordinate, while $c_{3, i}$ is distinct from the other two. The key point is that an equality pattern of the form $\{1, 2\}, \{3\}$ appears with probability $\approx \frac{1}{q}$, the equality pattern $\{1, 2, 3\}$ appears with probability $\approx \frac{1}{q^2}$, and the most likely equality pattern is $\{1\}, \{2\}, \{3\}$. Using this, one can show that for a fixed equality pattern over all $n$ coordinates, if the collision weight of that pattern is $w$, then the probability of that equality pattern appearing is $\leq q^{-w}$.

We can then observe that there are at most $5^n$ possible equality patterns ($5$ per coordinate),
and thus 
\[
\Pr[\wt(c_1, c_2, c_3) \geq t_2] \leq 5^n \cdot q^{-t_2}.
\]
The key fact now is that this allows us to bound the \emph{expected number} of violating triples in our randomly sampled code. Indeed, there are at most $\binom{2q^k}{3} \leq (2q^k)^3$ many triples of codewords in the sampled random code. Thus, the expected number of bad triples is bounded by 
\[
(2q^k)^3 \cdot q^{-t_2} \leq (2q^k)^3 \cdot 5^{n} \cdot q^{-2k} \cdot q^{-\frac{O(n)}{\log(q)}} \ll q^k.
\]
To conclude the proof then, it suffices to observe that for each such bad triple, we can delete a single participating bad codeword. This removes all bad triples from the code, but only results in the deletion of $\ll q^k$ many codewords. Hence, there are still $\geq q^k$ many codewords left in the code, as we started with $2q^k$ many codewords! All together, this then yields a code with $\geq q^k$ many codewords and no bad triples, which by our aforementioned discussion, is thus $(\rho, 2, 2)$ list recoverable.

\subsection{AEL Impossibility}\label{subsec:AEL-sketch}

Now we sketch the proof ideas of \cref{cor:main}. It is a parameterized version of \cref{thm:ael-main}, which says for any $\ell,m\ge 2$, when (1) the outer code $C_{\out}$ satisfies certain conditions, (2) the inner code encoder can be selected adversarially, and (3) the normalized second eigenvalue of the shuffle expander is small enough, the AEL code is not guaranteed to be $(1-\frac{m(R-\eps)}{m-1},\ell,L)$ list recoverable where $R$ denotes the rate of inner code and $L=\ell^m-1$. Specifically, we can find $\ell^m$ bad codewords in the AEL code that form an $m$-dimensional hypercube such that: On a set $S\subseteq[n]$ of coordinates with size $|S|=\frac{m(R-\eps)n}{m-1}$, for each $i\in S$, the $\ell^m$ bad codewords have only $\ell$ possibilities on the entry $i$.  This hypercube has side-length $\ell-1$, so the $\ell^m$ bad codewords can be seen as points in this unit-distance $m$-dimensional grid. This construction is inspired by \cite{CZ25}. 

Let us fix the simplest example $\ell=2$ and $m=2$, where many key ideas in the proof already show up. In this case our goal is to find four distinct codewords $\{c_{a,b}\colon a,b\in\{0,1\}\}$ in the AEL code such that for disjoint $I_a,I_b\subseteq[n]$ with $|I_a|=|I_b|=(R-\eps)n$, we have for any $a\in\{0,1\}$, $c_{a,0}[I_a]=c_{a,1}[I_a]$, and for any $b\in\{0,1\}$, $c_{0,b}[I_b]=c_{1,b}[I_b]$. If this holds, then for any $i\in I_a\cup I_b$, we know $|\{c_{a,b}[i]\colon a,b\in\{0,1\}\}|\leq 2=\ell$, which means there is a size-$\ell$ input list on $i$-th coordinate such that all the four codewords agree on it.

Since each codeword of the AEL code comes from some codeword of the outer code, we should start by finding the outer codewords $\{c'_{a,b},a,b\in\{0,1\}\}$ that will be then transformed into $\{c_{a,b}\colon a,b\in\{0,1\}\}$ after concatenation with inner code and expander reshuffle. We will take advantage of the large rate $R_{\out}$ of outer codes to find $\{c'_{a,b},a,b\in\{0,1\}\}$, and then use the spectral expansion and a carefully chosen encoders for inner codes to transform them into the final list $\{c_{a,b},a,b\in\{0,1\}\}$\
Our requirement for these $\{c'_{a,b}\}$ is the following: For a partition $J_0\cup J_a\cup J_b$ of $[n]$ that will be chosen later, there should be
\begin{itemize}

\item The four outer codewords $\{c'_{a,b}\}$ agree with each other on $J_0$.
\item For any $a\in\{0,1\}$, $c'_{a,0}[J_a]=c'_{a,1}[J_a]$.
\item For any $b\in\{0,1\}$, $c'_{0,b}[J_b]=c'_{1,b}[J_b]$.
\end{itemize}
Just this task of finding codewords satisfying the above is already very subtle. For this technical overview, we focus primarily on conveying the intuition for how we prove these codewords exist. 

Note that this is a stronger requirement than the conditions required for $\{c_{a,b}\}$ since we need all four codewords to agree on $J_0=[n]\backslash(J_a\cup J_b)$. However, we can find satisfying codewords $\{c'_{a,b}\}$ if the rate of outer code is large enough as follows.
\paragraph{If $C_{\out}$ is additive and $R_{\out}>\frac{m-1}{m}=\frac{1}{2}$.} Since $C_{\out}$ is additive, then for any $S\subseteq [n]$ with $|S|<R_{\out}n$, there exists $c'_S\in C_{\out}$ that vanishes on $S$ ($c'_S[S]=0$) but non-zero on remaining coordinates ($c'_{S}[[n]\backslash S]\neq 0$). Let $J_0\cup J_a\cup J_b$ be any partition such that $|J_0|=(1-\frac{2}{m})n$, $|J_a|=|J_b|=\frac{n}{m}$, then we know $|J_0\cup J_a|=|J_0\cup J_b|<R_{\out}n$. In this case, we can just define $c'_{a,b}=ac'_{J_0\cup J_b}+bc'_{J_0\cup J_a}$ for any $a,b\in\{0,1\}$, and it is easy to verify that these four codewords satisfy the required conditions. The formal statement of this step can be found in \cref{lem:outer-list}.

\paragraph{If $C_{out}$ is non-additive and $R_{\out}>1-\frac{1}{2m\ell^{m-1}}=\frac{7}{8}$.} In this case we use a generalization of the famous Kővári–Sós–Turán Theorem (\Cref{thm:erdos}). Suppose $|J_a|=|J_b|=s,|J_0|=n-2s$, let $q=|\Sigma_{\out}|$,  since $|C_{\out}|=q^{R_{\out}n}$, by pigeonhole principle there must be $C'\subseteq C_{\out}$, $|C'|=q^{R_{\out}n-n+2s}$ so that all codewords in $C'$ agree on $J_0$. Then we construct a bipartite graph $G=(L\cup R,E)$ so that $|L|=|R|=q^s$ and $|E|=|C'|> q^{R_{\out}n-n+2s}$. The construction of $G$ is simple: We interpret vertices in $L$ and $R$ both as vectors in $\Sigma_{\out}^s$. Then, for each $c'\in C'$, we add an edge $(c'[J_a],c'[J_b])$ to the graph. The Kővári–Sós–Turán Theorem tell use that if $|E|=\omega(|L|^{3/2})$, there must be a $K_{2,2}$ subgraph in $G$. When $R_{\out}$ is large enough, we can set the parameter $s$ accordingly so that the condition holds. Then, we can make the four codewords that consist of the four edges of the $K_{2,2}$ subgraph of $G$ as the desired $\{c'_{a,b}\}$. By the definition of $G$, we know that these codewords agree on $J_0$ and $|\{c'_{a,b}[i]\colon a,b\in\{0,1\}\}|=2$ for each $i\in J_a\cup J_b$, so these codewords satisfy our conditions.

Then, we will construct $\{c_{a,b}\}$ from $\{c'_{a,b}\}$ by choosing specific inner code encoders. For each $i\in[n]$, fix the disjoint right vertex sets $I_a,I_b\subseteq [n]$ that we wish to construct agreement, and then for each left vertex $i\in[n]$, there are only three cases:
\begin{itemize}
\item If $i\in J_0$, then since the four codewords already agree on $i$, we do not need to do anything special
\item If $i\in J_a$, then let $N_{i,a},N_{i,b}\subseteq [D]$ denote the set of edges between $i$ and $I_a$, $I_b$ respectively. Edges in $N_{i,a}$ only require $c_{a,0}=c_{a,1}$ for any $a\in\{0,1\}$. This is already guaranteed by the choices of $\{c'_{a,b}\}$. The new challenge is that for edges in $N_{i,b}$, they need $c_{0,b}=c_{1,b}$ for any $b\in\{0,1\}$. This means we have to choose the inner code encoder $\Enc_{i}\colon \Sigma_{\out}\to C_{\mathsf{in}}$ so that $\Enc_{i}(c'_{0,b}[i])[N_{i,b}]=\Enc_{i}(c'_{1,b}[i])[N_{i,b}]$. If $|N_{i,b}|<R_{\mathsf{in}}D$, then by pigeonhole principle we can guarantee there exists such an encoder and provide the sufficient agreement on $N_{i,b}$. If the graph $G$ is a random $D$-regular graph, then $\mathbb{E}[|N_{i,b}|]=\frac{|I_b|}{n}=(R_{\mathsf{in}}-\eps)D$. Since the spectral expander $G$ is pseudorandom graph with small enough second eigenvalue $\lambda_2/D$, by the expander mixing lemma (\Cref{lem:expander}) we know that the condition $|N_{i,b}|<R_{\mathsf{in}}D$ holds for \emph{almost all} left vertices. Therefore, we only need to select $J_a$ as a subset of ``good left vertices'', and the required inner encoder must exist.
\item The case $i\in J_b$ is symmetric to the former case.
\end{itemize}
This construction can generalize to any $m,\ell\ge 2$ and construct $\ell^m$ bad codewords. Therefore, we get the lower bound that the AEL frameworks can never be $(1-\frac{m(R-\eps)}{m-1},\ell,\ell^m-1)$ list recoverable for any  integers $m,\ell\ge2$. The main theorem is then derived from this bound as in \Cref{cor:main}.   

\subsection*{Acknowledgments}

We thank Mahdi Cheraghchi, Manik Dhar, Sivakanth Gopi, and Venkatesan Guruswami for many helpful discussions. 

\subsection*{AI Disclosure}

In the existence portion of \cref{thm:GSB-intro}, after realizing that the probabilistic method would be useful toward proving the statement, GPT 5.4 Pro was used to work out a detailed argument. The proof of \cref{thm:GSB-intro} presented in this manuscript was written entirely from scratch. GPT 5.5 Pro was also used to convert a hand-drawn version of \cref{fig:GSBupper} to TikZ.

\section{A Tight Generalized Singleton Bound for List Recovery}\label{sec:LR-GSB}

In this section, we prove an (essentially) tight generalized Singleton bound for list recovery. We begin with some preliminaries on set partitions and agreement hypergraphs.

\subsection{Partitions and Agreement Hypergraphs}

Given a finite set $S$, a partition of $S$ is a family of sets $\{S_1, \hdots, S_k\}$ such that $S_1, \hdots, S_k$ are pairwise disjoint and $S_1 \cup \cdots \cup S_k = S$. We write $S_1 \sqcup \cdots \sqcup S_k = S$ to succinctly denote that $\{S_1, \hdots, S_k\}$ is a partition of $S$.  We let $\mathcal P_S$ denote the set of all partitions of $S$. If $S = [m]$, then we use the notation $\mathcal P_m$ to denote $\mathcal P_{[m]}$.

Given a partition $P \in \mathcal P_S$, we have a corresponding equivalence relation $\sim_P$ on $S$ such that for any $i,j \in S$ we have that $i \sim_P j$ if and only if $i$ and $j$ are in the same part of $P$. We make the following simple observation.

\begin{proposition}
For all finite sets $S$, we have that $|\mathcal P_S| \le |S|!$.
\end{proposition}
\begin{proof}
Without loss of generality, we may assume that $S = [m]$ for some nonnegative integer $m$. We construct a partition $P \in \mathcal P_m$ iteratively as follows. Start with the empty partition. Then, for each $i \in \{1, 2, \hdots, m\}$, we can either choose that $i \sim_P j$ for some $j \in [i-1]$ or $i$ forms a new set in the partition. Some of these choices may be equivalent, but there are always at most $i$ choices. From this, we get an upper bound of $1 \cdot 2 \cdots m = m!$ on the number of partitions.
\end{proof}

Given, $P, Q \in \mathcal P_S$, we say that $P$ is a \emph{refinement} of $Q$, denoted by $P \preceq Q$ if for any $i, j \in S$ we have that $i \sim_P j$ implies $i \sim_Q j$. Conversely, we say that $Q$ is a \emph{coarsening} of $P$. Note that the partition of $S$ where every set is a singleton is the ``finest'' possible partition, and the of $S$ consisting of a single set is the ``coarsest'' possible partition.

Given a finite set $S$, We define an \emph{agreement hypergraph}~\cite{st20,AGGLZ25,CZ25} over $S$ to be any sequence $H = (P_1, \hdots, P_n) \in \mathcal P_S^n$. Given an alphabet $\Sigma$ and any codewords $c_1, \hdots, c_{m} \in \Sigma^n$, we can define their agreement hypergraph $\mathcal H(c_1, \hdots, c_m)$ to be a sequence $(P_1, \hdots, P_n) \in \mathcal P_{m}^n$ such that for all $i \in [n]$, we have that $P_i$ corresponds to the partition of $[m]$ such that $j,j' \in [m]$ are in the same set if and only if $c_{i,j} = c_{i,j'}$. Given $H \in \mathcal P_m^n$, we say that $(c_1, \hdots, c_m)$ \emph{certify} $H$ if $\mathcal H(c_1, \hdots, c_m) \preceq H$.

Given a partition $P \in \mathcal P_S$, we define its cardinality $|P|$ to be the number of sets in the partition and we define the partition's \emph{weight} to be $\wt(P) = |S| - |P|$. Then, for an agreement hypergraph $H = (P_1, \hdots, P_n) \in \mathcal P_S^n$ we define its cardinality to be
\[
    |H| = \sum_{i=1}^n |P_i|,
\]
and its weight tobe
\[
    \wt(H) = \sum_{i=1}^n \wt(P_i) = n|S| - |H|. 
\]
For any codewords $c_1, \hdots, c_m \in \Sigma^n$, we define
\[
    \wt(c_1, \hdots, c_m) = \wt(H(c_1, \hdots, c_m)) = \sum_{i=1}^n \left(m - |\{c_{1,i}, \hdots, c_{m,i}\}|\right).
\]
Given $P \in \mathcal P_S$ and a set $T \subseteq S$, we let $P[T] \in \mathcal P_S$ denote the partition on $T$ induced by $P$. Given $H = (P_1, \hdots, P_n) \in \mathcal P_S^n$, we let $H[T] := (P_1[T], \hdots, P_n[T])$.

\subsection{A Novel Generalized Singleton Bound for List Recovery}\label{subsec:GSB}

Generalizing \cite[Theorem 5.1]{GST24}, we can prove the following generalized Singleton bound (GSB) for agreement hypergraphs. Throughout, we let $\Sigma$ denote an alphabet of size $q$.

\begin{theorem}\label{thm:LR-GSB}
Let $L \ge \ell \ge 1$ and $n \ge 1$ be integers. Let $R, \rho \in [0, 1]$ be such that
\begin{align}
    \rho &\ge \frac{L+1-\ell}{L+1} - \frac{L}{L+1} \cdot R,\label{eq:LR-GSB}
\end{align}
or equivalently
\begin{align}
    R \ge \frac{L+1-\ell}{L} - \frac{L+1}{L}\rho\label{eq:LR-GSB-rate}
\end{align}
Then, for any $C \subseteq \Sigma^n$ of size at least $q^{Rn + 5L}$, we have that $C$ is \emph{not} $(\rho, \ell, L)$ list-recoverable.
\end{theorem}

\begin{remark}\label{rem:GST}
Theorem 5.1 of \cite{GST24} proves the an analogoue of \cref{thm:LR-GSB} with the weaker condition that 
\begin{align}
R \le 1 - \frac{L+1}{L+1-\ell}\rho\label{eq:LR-GST-rate}
\end{align}
or equivalently
\begin{align}
\rho \ge \frac{L+1-\ell}{L+1}(1-R) = \frac{L+1-\ell}{L+1} - \frac{L+1-\ell}{L+1} \cdot R\label{eq:LR-GST}
\end{align}
In particular, for all rates $R \in (0, 1)$, \cref{eq:LR-GSB} gives a strictly superior bound to \cref{eq:LR-GST} whenever $\ell \ge 2$, where the case $\ell=1$ corresponds to the list-decoding setting already established in \cite{st20}.

The most important qualitative difference between \cref{eq:LR-GSB-rate} and \cref{eq:LR-GST-rate} is that our bound has the ability to rule out zero-error $(0, \ell, L)$ list recovery for certain rates. For example, if $\ell = L = 2$, we can infer from \cref{eq:LR-GSB-rate} that $R \le \frac{1}{2}$ in order for our code $C$ to have a chance of being $(0, 2, 2)$ list-recoverable, whereas \cref{eq:LR-GST-rate} imposes the trivial constraint $R \le 1$. We are not the first to observe that there is are rate restrictions on zero-error list recovery, see \cite{RYZ24,rv26}. However, by virtue of \cref{cor:LR-GSB-tight}, we are the to precisely calculate the zero-error rate threshold over large alphabets. 
\end{remark}

\begin{proof}
Our goal is to find sets $S_1, \hdots, S_n \subseteq \Sigma$ of size at most $\ell$ and distinct codewords $c_1, \hdots, c_{L+1} \in C$ such that for all $i \in [L+1]$,
\begin{align}
    \sum_{j=1}^n \mathbf{1}[c_{i,j} \in S_j] \ge (1-\rho)n.\label{eq:GSB-goal}
\end{align}
Given a set $I \subseteq [n]$ of coordinates and a codeword $c \in C$, we let $\mathcal N_I(c) := \{c' \in C \mid c'|_{I} = c|_I\}$. For any integer $a \ge 1$, we further define 
\[
\mathcal F_{I,a} := \{c \in C \mid |N_I(c)| \le a\}.
\]
We observe the following fact.
\begin{claim}\label{claim:F-bound}
$|\mathcal F_{I,a}| \le a \cdot q^{|I|}$.
\end{claim}
\begin{proof}
If $|\mathcal F_{I,a}| > a \cdot q^{|I|}$, then by the pigeonhole principle, there is a string $x \in \Sigma^I$ such that there are at least $a+1$ choice of $c \in \mathcal F_{I,a}$ for which $c|_{I} = x$. For any one such $c$, we thus have that $|\mathcal N_I(c)| \ge a+1$, a contradiction of the definition of $\mathcal F_{I,a}$.
\end{proof}

Pick sets $I_1, \hdots, I_L \subseteq [n]$ subject to the following constraints.
\begin{enumerate}
\item $I_1, \hdots, I_L$ are disjoint.
\item $|I_j| \le \lfloor \frac{Rn}{L+1-\ell}\rfloor + 3$ for all $j \in [L]$.
\item $I_1 \cup \cdots \cup I_L$ is as large as possible.
\end{enumerate}
In particular, if $(\lfloor \frac{Rn}{L+1-\ell}\rfloor + 3) \cdot L \le n$, then $|I_j| =  \lfloor \frac{Rn}{L+1-\ell}\rfloor + 3$ for all $j \in [L]$, and otherwise $I_1 \cup \cdots \cup I_L = [n]$. We also let $I_0 := [n] \setminus (I_1 \cup \cdots \cup I_L)$. For all $j \in \{2, \hdots, L+1\}$, we define
\[
    J_j := I_{j-1} \cup I_{j\!\!\!\mod L} \cup I_{j+1\!\!\!\mod L} \cup \cdots \cup I_{j+L-\ell-1\!\!\!\mod L}.
\]
That is, $J_j$ is a union of $L+1-\ell$ sets. By \cref{claim:F-bound}, we have that
\[
    \left|\bigcup_{j=2}^{L+1} \mathcal F_{J_j, L}\right| \le L^2 \max_{j \in \{2, \hdots, L+1\}} q^{|J_j|} \le L^2 q^{(L+1-\ell) (\lfloor \frac{Rn}{L+1-\ell}\rfloor + 3)} \le L^2 q^{Rn + 3(L+1-\ell)} < q^{Rn + 5L} \le |C|.
\]
Therefore, we can find $c_1 \in C$ which is not contained in $\mathcal F_{J_j,L+1}$ for any $j \in \{2, \hdots, L+1\}$. With this choice of $c_1$, observe that one can greedily pick $c_2 \in \mathcal N_{J_2}(c_1), \hdots, c_{L+1} \in \mathcal N_{J_{L+1}}(c_1)$ which are distinct. The reason is that by choice of $c_1$, we have that $|N_{J_j}(c_1)| \ge L+1$ for all $j \in \{2, \hdots, L+1\}$, so each $c_j$ can be chosen distinctly from the others.

Observe that for any $j \in [L]$ there are exactly $L+1-\ell$ choices of $j' \in \{2, \hdots, L+1\}$ for which $I_J \subseteq J_{J'}$. In particular, this means that for any index $i \in I_1 \cup \cdots \cup I_L$, we have that some $L+1-\ell$ of $c_{1,i}, \hdots, c_{L+1,i}$ are identical. Thus, for all $i \in I_1 \cup \cdots \cup I_L$, we can define $S_i := \{c_{1,i}, \hdots, c_{L+1,i}\}$ and know that $|S_i| \le \ell$.

If $I_0 = \emptyset$, then our choice of $c_1, \hdots, c_{L+1}$ already certifies that $C$ is not $(0, L+1, \ell)$-list recoverable and thus cannot be $(\rho, L+1, \ell)$-list recoverable. Now assume that $I_0 \neq \emptyset$ so $|I_j| =  \lfloor \frac{Rn}{L+1-\ell}\rfloor + 3$ for all $j \in [L]$. For all $i_0 \in I_0$, define $S_{i_0}$ to be an arbitrary subset of $[L+1]$ of size $\ell$ such that for all $i \in [L+1]$, we have the balancing condition that
\[
    \sum_{j \in I_0} \mathbf{1}[c_{i,j} \in S_j] \ge \left \lfloor\frac{|I_0| \ell}{L+1}\right \rfloor.
\]
\allowdisplaybreaks
In particular, we have that for all $i \in [L+1]$ that
\begin{align*}
    \sum_{j=1}^n \mathbf{1}[c_{i,j} \in S_j] &\ge n - |I_0| + \left \lfloor\frac{|I_0| \ell}{L+1}\right \rfloor\\
    &\ge n - 1- |I_0| \cdot \frac{L+1-\ell}{L+1} \\
    &= n - 1 - \left(n - L \left(\left\lfloor \frac{Rn}{L+1-\ell}\right\rfloor + 3 \right)\right) \cdot \frac{L+1-\ell}{L+1}\\
    &\ge  n - 1 - \left(n - L \left(\frac{Rn}{L+1-\ell} + 2\right)\right) \cdot \frac{L+1-\ell}{L+1}\\
    &= \frac{\ell}{L+1} \cdot n + \frac{L}{L+1} \cdot Rn + 2L \cdot \frac{L+1-\ell}{L+1}-1\\
    & \ge \left(1 - \left[\frac{L+1-\ell}{L+1} - \frac{L}{L+1} \cdot R\right]\right) \cdot n\\
    & \ge (1 - \rho) n,
\end{align*}
which proves \cref{eq:GSB-goal}. Thus, $C$ is not $(\rho, L+1, \ell)$-list recoverable.
\end{proof}

\begin{remark}
The proof has similarities to many other GSB proofs in the literature~\cite{st20,GST24,CZ25,LMS25,LS25}, particularly in the use of the pigeonhole principle\footnote{Or in the case of \cite{CZ25,LMS25,LS25} a dimension-counting argument that could be viewed as a linear-algebraic variant of pigeonhole} to identify suitable codewords $c_1, \hdots, c_{L+1}$. In comparison to \cite{GST24}, our use of pigeonhole is superior in the sense that we impose the information-theoretic maximum possible number of constraints on $c_1, \hdots, c_{L+1}$.
\end{remark}

\subsection{Tightness of List Recovery Generalized Singleton Bound}

We now prove that \cref{thm:LR-GSB} is essentially tight by showing that codes exist (over sufficiently large alphabets) which match \cref{eq:LR-GSB}. To begin, we use the probabilistic method to show there exists a code $C$ with a much more general property.

\begin{theorem}\label{thm:LR-GSB-constr}
For any $n > k \ge 1$ and $q = |\Sigma| > L_0 \ge 1$, there exists $C \subseteq \Sigma^n$ of size $q^k$ such that for all $L \in [L_0]$ and distinct $c_1, \hdots, c_{L+1} \in C$ we have that
\begin{align}
\wt(c_1, \hdots, c_{L+1}) \le Lk + \frac{(L+1)\log(L+1)}{\log q} (n+3).\label{eq:GSB-UB}
\end{align}
\end{theorem}

The proof proceeds by using the probabilistic method~\cite{AS08}, with some technical similarities to recent result for list-decoding and local properties \cite{AGL25,CZ25,LMS25} (particularly \cite[Proposition I.4]{AGL25}). We make crucial use of the following lemma.

\begin{lemma}\label{lem:prob-bound}
Let $c_1, \hdots, c_{L+1} \sim \Sigma^n$ be sampled uniformly at random. For any $t \ge 0$, with probability at most $((L+1)!)^n q^{-t}$, we have that $\wt(c_1, \hdots, c_{L+1}) \ge t.$
\end{lemma}
\begin{proof}
Recall that $\mathcal H(c_1, \hdots, c_{L+1}) \subseteq \mathcal P_{L+1}^n$ is the agreement hypergraph induced by the codewords $c_1, \hdots, c_{L+1}$. Observe that
\begin{align*}
\wt(c_1, \hdots, c_{L+1}) &= \sum_{i=1}^n \left(L+1 - |\{c_{1,i}, \hdots, c_{L+1,i}\}|\right)\\
    &= \sum_{i=1}^n \left(L+1 - |\mathcal H_i(c_1, \hdots, c_{L+1})|\right)\\
    &= \wt(\mathcal H(c_1, \hdots, c_{L+1}).
\end{align*}
Now fix an agreement hypergraph $\bar{P} = (P_1, \hdots, P_n) \in \mathcal P_{L+1}^n$. We claim that for a uniformly random choice of $c_1, \hdots, c_{L+1} \sim \Sigma^n$ that $\Pr[\mathcal H(c_1, \hdots, c_{L+1}) = \bar{P}] \le q^{-\wt(\bar{P})}$. To see why, by independence of the sampled coordinates it suffices to prove for each $i \in [n]$ that $\Pr[\mathcal H(c_1, \hdots, c_{L+1})_i = P_i] \le q^{|P_i|-(L+1)}$. Let $S_1, \hdots, S_{|P_i|}$ be the constituent sets of $P_i$ and let $s_1 \in S_1, \hdots, s_{|P_i|} \in S_{|P_i|}$ be arbitrary representatives. For every $j \in [L+1] \setminus \{s_1, \hdots, s_{|P_i|}\}|$, we require that $c_{j,i} = c_{s_{j^*},i}$ where $j^*$ is chosen such that $j \in S_{j^*}$. This is a list of $L+1 - |P_i|$ independent events that must hold and each event happens with probability $1/q$. Therefore, $\Pr[\mathcal H(c_1, \hdots, c_{L+1})_i = P_i] \le q^{|P_i|-(L+1)}$, so $\Pr[\mathcal H(c_1, \hdots, c_{L+1}) = \bar{P}] \le q^{-\wt(\bar{P})}$.

To finish, we observe that for uniformly random $c_1, \hdots, c_{L+1} \in \Sigma$ we have that
\begin{align*}
    \Pr[\wt(c_1, \hdots, c_{L+1}) \ge t] &= \sum_{\substack{\bar{P} \in \mathcal P_{L+1}^n\\\wt(\bar{P}) \ge t}} \Pr[\mathcal H(c_1, \hdots, c_{L+1}) = \bar{P}]\\
    &\le \sum_{\substack{\bar{P} \in \mathcal P_{L+1}^n\\\wt(\bar{P}) \ge t}} q^{-\wt(\bar{P})}\\
    &\le |P_{L+1}^n| q^{-t} = ((L+1)!)^n q^{-t}.
\end{align*}
From this bound, the lemma follows.
\end{proof}

We now proceed to prove \cref{thm:LR-GSB-constr}.

\begin{proof}[Proof of \cref{thm:LR-GSB-constr}]
For any $L \in [L_0]$ define
\[
t_L = Lk + \frac{(L+1)\log(L+1)}{\log q} (n+3).
\]
Consider the case in which $t_1 \ge n$. Observe then for all $L \in L_0$ that
\begin{align*}
    t_L - L t_1 &= Lk + \frac{(L+1)\log(L+1)}{\log q} (n+3) - L(k + \frac{2 \log 2}{\log q} (n+3))\\
    &= \frac{(L+1)\log (L+1) - 2L \log 2}{\log q} (n+3) \ge 0,
\end{align*}
where the inequality follows from that fact that $(L+1)^{L+1} \ge 2^{2L}$ for all $L \ge 1$. Therefore, $t_L \ge Ln$ for all $L \in [L_0]$. This implies that the conditions \cref{eq:GSB-UB} are all vacuous, so any $C \subseteq \Sigma^n$ of size $q^k$ suffices.

Hence, we now assume that $t_1 < n$. Sample codewords $c_{1}, \hdots, c_{2q^{k}} \sim \Sigma^n$ uniformly at random. 
and define the random variable $X_L$ to be the set of all sequences $\{i_1, \hdots, i_{L+1}\} \in \binom{[2q^k]}{L+1}$ with $\wt(c_{i_1}, \hdots, c_{i_{L+1}}) \ge t_L$. By \Cref{lem:prob-bound} the probability that any individual $(i_1, \hdots, i_{L+1})$ is an element of $X_L$ is at most $((L+1)!)^n q^{-t_L}$. As such,
\begin{align*}
    \mathbb E[|X_L|] &\le \binom{2q^k}{L+1} ((L+1)!)^n q^{-t}\\
      &\le (2q^k)^{L+1} (L+1)^{n(L+1)} q^{-t}\\
      &= q^{(L+1)k - t_L} \cdot 2^{L+1}(L+1)^{n(L+1)}\\
      &= q^k \cdot (L+1)^{-(L+1)(n+3)} \cdot 2^{L+1}(L+1)^{n(L+1)}\\
      &= q^k (L+1)^{-3(L+1)} \cdot 2^{L+1}\\
      &\le q^k 2^{-2(L+1)}.
\end{align*}
Thus,
\begin{align*}
    \mathbb E\left[\sum_{L=1}^{L_0} |X_L|\right] &\le \sum_{L=1}^{L_0} q^k 2^{-2(L+1)}
    \le q^k \sum_{L=1}^{L_0} 2^{-2(L+1)}
    \le q^k.
\end{align*}
Hence, there exists $c_{1}, \hdots, c_{2q^{k}} \in \Sigma^n$ such that $\sum_{L=1}^{L_0} |X_L| \le q^k$. For each $L \in [L_0]$ and $(i_1, \hdots, i_{L_0}) \in X_L$, delete $c_{i_1}$ from the list $c_{1}, \hdots, c_{2q^{k}}$. After doing this, we may pick a subsequence $c'_1, \hdots, c'_{q^k} \in \Sigma^n$ such that for any $L \in [L_0]$ and any distinct $i_1, \hdots, i_{L+1} \in [q^k]$ we have that $\wt(c'_{i_1}, \hdots, c'_{L+1}) \le t_L$.

To finish, we need to check that $c'_{1}, \hdots, c'_{q^k}$ are distinct. Recall we assume that $t_1 < n$. Thus, for any $i, j \in [q^k]$, we have that $\wt(c'_{i}, c'_j) < n$. This can only occur if $c'_i \neq c'_j$ for all $i, j \in [q^k].$ Therefore, $C = \{c'_1, \hdots, c'_{q^k}\}$ is our desired code of size $q^k$.
\end{proof}

In fact, the proof of \cref{thm:LR-GSB-constr} also shows that the following stronger condition is true: 
\begin{corollary}\label{cor:LR-GSB-constr}
For any $n > k \ge 1$ and $q = |\Sigma| > L_0 \ge 1$, a random code $C \subseteq \Sigma^n$ of size $2q^k$ contains (with probability $\geq 7/8$) a subcode $C' \subseteq C$ of size $q^k$ such that for all $L \in [L_0]$ and distinct $c_1, \hdots, c_{L+1} \in C'$ we have that
\begin{align}
\wt(c_1, \hdots, c_{L+1}) \le Lk + \frac{(L+1)\log(L+1)}{\log q} (n+3).
\end{align}
\end{corollary}

\begin{proof}[Proof of \cref{cor:LR-GSB-constr}.]
We use the same notation as in the proof of \cref{thm:LR-GSB-constr}. In particular, we see that 
\[
\mathbb E\left[\sum_{L=1}^{L_0} |X_L|\right] \le \sum_{L=1}^{L_0} q^k 2^{-2(L+1)}
    \le q^k \sum_{L=1}^{L_0} 2^{-2(L+1)}
    \le \frac{q^k}{8}.
\]
By a simple Markov bound, this then implies that with probability $\geq 7/8$ that $\sum_{L=1}^{L_0} |X_L| \leq q^k$. Conditioned on this being true, we can then repeat the process of \cref{thm:LR-GSB-constr}: for each $L \in [L_0]$ and $(i_1, \hdots, i_{L_0}) \in X_L$, delete $c_{i_1}$ from the list $c_{1}, \hdots, c_{2q^{k}}$. After doing this, we may pick a subsequence $c'_1, \hdots, c'_{q^k} \in \Sigma^n$ such that for any $L \in [L_0]$ and any distinct $i_1, \hdots, i_{L+1} \in [q^k]$ we have that $\wt(c'_{i_1}, \hdots, c'_{L+1}) \le t_L$.
\end{proof}

Our main application of \cref{thm:LR-GSB-constr} however is that \cref{thm:LR-GSB} is asymptotically tight as $n,q$ tend to infinity.

\begin{corollary}\label{cor:LR-GSB-tight}
For any $n \ge 2$, $q = |\Sigma| > L \ge \ell \ge 1$, and $R \in [0, 1]$, pick $\rho \in [0, 1]$ such that
\[
    \rho \le \frac{L+1-\ell}{L+1} - \frac{L}{L+1} \cdot R.
\]
Let \[
k = \left\lfloor Rn - \frac{(L+1)\log(L+1)}{\log q}(n+3)-1\right\rfloor
\]
and assume $k\ge 1$, then there exists $C \subseteq \Sigma^n$ of size $q^{k}$ such that $C$ is $(\rho, \ell, L)$ list recoverable.
\end{corollary}

\begin{proof}
Apply \cref{thm:LR-GSB-constr} with $L_0 = L$ to construct a code $C \subseteq \Sigma^n$ of size $q^{k}$ such that for any distinct $c_1, \hdots, c_{L+1} \in C$ we have that
\begin{align}
\wt(c_1, \hdots, c_{L+1}) \le Lk + \frac{(L+1)\log(L+1)}{\log q} (n+3).\label{eq:xz}
\end{align}
Now, assume for sake of contradiction that $C$ is not $(\rho, \ell, L)$ list recoverable. Then, there exists distinct $c_1, \hdots, c_{L+1} \in C$ as well as sets $S_1, \hdots, S_\ell \subseteq \Sigma$ of size $\ell$ such that for all $i \in [L+1]$,
\[
    \sum_{j=1}^n \mathbf{1}[c_{i,j} \in S_j] \ge (1-\rho)n.   
\]
As such, we have that
\begin{align*}
    \wt(c_1, \hdots, c_{L+1}) &= \sum_{j=1}^n \left(L+1 - |\{c_{1,j}, \hdots, c_{L+1,j}\}|\right)\\
    &\ge \sum_{j=1}^n \left[L+1 - |S_j| - \sum_{i=1}^{L+1} \mathbf{1}[c_{i,j} \not\in S_j]\right]\\
    &\ge -n\ell + \sum_{i=1}^{L+1}\sum_{j=1}^n \mathbf{1}[c_{i,j} \in S_j]\\
    &\ge ((L+1)(1-\rho) - \ell)n\\
    &\ge \left((L+1)\left(1 - \frac{L+1-\ell}{L+1} + \frac{L}{L+1} \cdot R\right) - \ell\right) n\\
    &= LRn.
\end{align*}
Therefore, by \cref{eq:xz} we have that
\begin{align*}
    LRn &\le Lk + \frac{(L+1)\log(L+1)}{\log q} (n+3)\\
    &= L\left\lfloor Rn - \frac{(L+1)\log(L+1)}{\log q}(n+3) - 1\right\rfloor + \frac{(L+1)\log(L+1)}{\log q} (n+3)\\
    &\le L \left( Rn - \frac{(L+1)\log(L+1)}{\log q}(n+3) - 1\right) + \frac{(L+1)\log(L+1)}{\log q} (n+3)\\
    &\le LRn - 1,
\end{align*}
a contradiction. Thus, $C$ is indeed $(\rho, \ell, L)$ list recoverable.
\end{proof}

In particular, this means that the RHS of \cref{eq:LR-GSB} is precisely the error threshold $\rho$ for which $(\rho, \ell, L)$ list recovery is possible.

\section{AEL Lower Bound}\label{sec:ael}
\subsection{Preliminaries}
We begin by stating a few standard facts we need concerning graphs and hypergraphs.

\begin{theorem}[Generalized Kővári–Sós–Turán Theorem \cite{erdos64}]\label{thm:erdos}
Given any $r,t>0$, there exists a constant $C>0$ such that for any $r$-uniform $r$-partite hypergraph $G$ on $rn$ vertices with each part size $n$, if the hypergraph has no $K^{(r)}_{t,\dots,t}$ as a subgraph, then $G$ has at most $Cn^{r-\frac{1}{t^{r-1}}}$ hyperedges.
\end{theorem}
\begin{definition}[Spectral Expansion]
Let $G=([n]\cup[n],E)$ be a $d$-regular bipartite graph, and $\lambda$ be the second largest eigenvalue of its adjacency matrix, then for any $\lambda'\ge \lambda/d$, we say $G$ is a $\lambda'$-spectral expander.
\end{definition}
\begin{theorem}[Expander Mixing Lemma]\label{lem:expander}
Let $G=([n]\cup[n],E)$ be a $d$-regular $\lambda$-spectral expander. For any $S,T\subseteq [n]$. Let $E(S,T)=E
\cap (S\times T)$, there is
\[
\left||E(S,T)|-\frac{d|S||T|}{n}\right|\leq \lambda dn.
\]

\end{theorem}

\paragraph{Inner-encoder insensitive AEL Framework.}

We now present the basic components of the AEL framework.

\begin{itemize}

\item Ingredients: An outer code $C_{\out}\subseteq \Sigma_{\out}^{n}$ with rate $R_{\out}$. $n$ inner codes $C_{\mathsf{in},1},\dots,C_{\mathsf{in,n}}\subseteq \Sigma_{\mathsf{in}}^D$, each with rate $R_{\mathsf{in}}$ and size $|\Sigma_{\out}|$. A 
$D$-regular $\lambda$-spectral expander $G=([n]\cup[n],E)$. 
\item Construction: Pick \emph{arbitrary} inner encoding maps $\mathsf{Enc}_{\mathsf{in},i}\colon \Sigma_{\mathsf{out}}\to \Sigma_{\mathsf{in}}^D$. We construct $C_{\mathsf{AEL}}\subseteq (\Sigma_{\mathsf{in}}^D)^n$ as follows:

For each $c_{\mathsf{out}}\in C_{\mathsf{out}}$, there is a codeword $c\in C_{\mathsf{AEL}}$ defined as follows: For each edge $e=(u,v)\in E$ which is the $e_u$-th edge connected to $u$, let the symbol on this edge to be $c_e=\mathrm{Enc}_{\mathsf{in},u}(c_{\mathsf{out}}[u])[e_u]$. Then, for each right vertex $v\in[n]$, let $e_{v,1},\dots,e_{v,D}\in E$ denote edges connected to $v$, we define $c$ as
\[
c=((c_{e_{1,1}},\dots,c_{e_{1,D}}),\dots,(c_{e_{n,1}},\dots,c_{e_{n,D}}))\in (\Sigma_{\mathsf{in}}^D)^n.
\]
\end{itemize}

\subsection{Output List Size Lower Bound}

\begin{theorem}\label{thm:ael-main}
Let $m,\ell\ge 2$ be positive constant integers, and $\eps,R_{\mathsf{in}},R_{\mathsf{out}}$ be constants. There exist constants $C_0>0$ and $\lambda>0$ such that for any inner-encoder insensitive AEL framework parameterized by  \[(n,C_{\out},R_{\mathsf{out}},C_{\mathsf{in},1},\dots,C_{\mathsf{in},n},R_{\mathsf{in}},D,\lambda,G)\] if at least one of the two following conditions holds:
\begin{itemize}
\item $R_{\mathsf{out}}\ge1-\frac{1}{2m\ell^{m-1}}$.
\item $\mathcal{P}_{\mathsf{out}}$ only contains additive codes and $R_{\out}>\frac{m-1}{m}$. 
\end{itemize}
then for any large enough $n> C_0$, the prescribed AEL framework does not guarantee to construct a $(1-\frac{m(R_{\mathsf{in}}-\eps)}{m-1},\ell,\ell^m-1)$ list-recoverable code assuming $R_{\mathsf{in}}\leq\frac{m-1}{m}$. 
\end{theorem}
\begin{proof}

We choose $s\in(0,\frac{1}{m})$ as follows:
\begin{itemize}
\item In case \(\rout\geq 1-1/(2m\ell^{m-1})\), set
\[
   s:=\frac{3}{4m}.
\]
Then
\[
   \rout-1+\frac{s}{\ell^{m-1}}
   \geq -\frac{1}{2m\ell^{m-1}}+\frac{3}{4m\ell^{m-1}}
   =\frac{1}{4m\ell^{m-1}}>0.
\]
\item In case $\mathcal{P}_{\rm out}$ only contains additive codes and \(\rout>(m-1)/m\), set
\[
   s:=\frac{1}{2}\left(1-\rout+\frac1m\right).
\]
Then \(1-\rout<s<1/m\), and hence
\[
   \rout-1+s>0.
\]
Choose constants $\eta<1-ms,\lambda\leq \min(\frac{3\eps\eta}{8m},\sqrt{\frac{\eps}{2\log_{2}{\ell}}})$.
\end{itemize}

Let $b=\frac{(R_{\mathsf{in}}-\eps)}{m-1}$, then since $R_{\mathsf{in}}\leq\frac{m-1}{m}$, we can select disjoint $I_1,\dots,I_m\subseteq [n]$, each with size $bn$. It suffices to showw that the AEL framework can construct a code $C_{\mathsf{AEL}}$ that contains distinct codewords $\{c_{d}\colon d\in [\ell]^m\}$ such that for any $i\in[m],d_1,d_2\in[\ell]^m$, there is $d_1[i]=d_2[i]\Rightarrow c_{d_1}[I_i]=c_{d_2}[I_i]$. Because if this condition holds, we know for any $j \in I_1\cup\cdots\cup I_m$, there is $|\{c_d[j]\colon d\in[\ell]^m\}|=\ell$. This would imply that the code is not $(\rho,\ell,\ell^m-1)$ list-recoverable where $\rho=1-\frac{1}{n}|I_1\cup\cdots\cup I_m|=1-mb=1-\frac{m(R_{\mathsf{in}}-\eps)}{m-1}$.

Let $I=I_1\cup \cdots\cup I_m$. For two subsets $L,R\subseteq [n]$, let $E(L,R)$ denote the set of edges between left vertices in $L$ and right vertices in $R$ in $G$. For each $i\in[m]$, we define the left vertex set $G_i$ as follows
\[
G_i=\{i\in[n]\colon |E(\{i\},I\setminus I_i)|\leq ((m-1)b+\frac{\eps}{2})D\}.
\]

By \Cref{lem:expander}, we have that
\[
   |{E([n]\setminus G_i,I\setminus I_i)}|
   \leq D\frac{|[n]\setminus G_i||I\setminus I_i|}{n}+\lambda Dn
   \leq D\left((m-1)b+\frac{\eps}{8}\right)|[n]\setminus G_i|+\lambda Dn.
\]
Also, since $ |{E([n]\setminus G_i,I\setminus I_i)}|>((m-1)b+\frac{\eps}{2})D|[n]\setminus G_i|$, we know that
\[
|[n]\setminus G_i|\leq \frac{8\lambda n}{3\eps}\leq \frac{\eta n}{m}.
\]
Therefore, we know $|G_1\cap\cdots\cap G_m|\ge (1-\eta)m$. Since $ms\leq 1-\eta$, we can choose disjoint $J_1,\dots,J_m\subseteq G_1\cap\cdots\cap G_m$, each with size $sn$ as subsets of left vertices. Let $J_0=[n]\setminus(J_1\cup\cdots\cup J_m)$. We start from find the bad list from the outer code $C_{\out}$ as follows.
\begin{lemma}\label{lem:outer-list}
There exist pairwise distinct $\{c'_{d}\colon d\in[\ell]^m\}\subseteq C_{\rm out}$ such that for any $i\in[m],d_1,d_2\in[\ell]^m$, there is $d_1[i]=d_2[i]\Rightarrow c'_{d_1}[J_i]=c'_{d_2}[J_i]$. Moreover, there is $c'_{d_1}[J_0]=c'_{d_2}[J_0]$ for any $d_1,d_2\in[\ell]^m$.
\end{lemma}
\begin{proof}
There are two cases, the first one is when $C_{\rm out}$ is additive and $R_{\rm out}-1+s>0$. Then since for each $i\in[m]$ there is $|[n]\setminus J_i|=(1-s)n$. Let $n>\frac{\log{\ell}}{R_{\out}-1+s}$ to be large enough. Since $C_{\out}$ is additive, let $C_{\out}$ be $\F_q$ linear and $m=\log_{q}|\Sigma_{\out}|$, $U\subseteq \F^{R_{\out}nm}$ denote the message space of $C_{\out}$, there must be a subspace $V\subseteq U$ with dimension at least $(R_{\out}-1+s)nm$ that contains all messages whose encodings are zero on $[n]\setminus J_i$. Therefore, there are at least $|\Sigma_{\out}|^{(R_{\rm out}-1+s)n}\ge2^{R_{\out}-1+s}\ge\ell$ distinct codewords $c\in C_{\rm out}$ such that $j\notin J_i\Rightarrow c[j]=0$. Pick $\ell$ of them to form a set $U_i=\{c''_{i,1},\dots,c''_{i,\ell}\}\subseteq C_{\rm out}$.

Now we construct the outer codeword set as
\[
c'_d=\sum^m_{i=1}c''_{i,d[i]}\colon \forall d\in[\ell]^m.
\]

Let us verify this set of $\{c'_d\colon d\in[\ell]^m\}$ satisfies the target code property: For each $i\in[m]$ and $d_1,d_2\in[\ell]^m$ such that $d_1[i]=d_2[i]$, we know that
\[
c'_{d_1}[J_i]-c'_{d_2}[J_i]=\sum^m_{t=1}c''_{t,d_1[t]}[J_i]-c''_{t,d_2[t]}[J_i]=c''_{i,d_1[i]}[J_i]-c''_{i,d_2[i]}[J_i]=0
\]
Since codewords in $U_i$ are pairwise distinct on $J_i$, 
this also tells us that if 
$d_1[i]\neq d_2[i]$, then $c'_{d_1}[J_i]\neq c'_{d_2}[J_i]$. This implies $c'_{d_1}\neq c'_{d_2}\Leftrightarrow d_1\neq d_2$, which proves that $|\{c'_d\colon d\in[\ell]^m\}|=\ell^m$.

Finally it is easy to verify that
\[
c'_{d_1}[J_0]-c'_{d_2}[J_0]=\sum^m_{t=1}c''_{t,d_1[t]}[J_0]-c''_{t,d_2[t]}[J_0]=0\colon \forall d_2,d_2\in[\ell]^m
\]

Then we prove the lemma in the second case that $R_{\rm out}-1+\frac{s}{\ell^{m-1}}>0$ and $C_{\out}$ is possibly non-additive. Let $q_{\out}=|\Sigma_{\out}|$

 Since \(C_{\rm out}\) has rate at least \(\rout\),
\[
   |{C_{\rm out}}|\geq q_{\rm out}^{\rout n}.
\]
Let the projection map be $g\colon C_{\out}\to\Sigma^{|J_0|}_{\out}$ defined by $g(c)=c[J_0]$, by pigeonhole principle we can choose the largest fiber of $g$, say
\(C'\subseteq C_{\rm out}\), such that $|g(C')|=1$ and there is
\[
   |{C'}|
   \geq \frac{|{C_{\rm out}}|}{q_{\rm out}^{|{J_0}|}}
   \geq q_{\rm out}^{(\rout-1+ms)n}.
\]
Construct an \(m\)-uniform \(m\)-partite hypergraph \(H\) whose \(h\)-th vertex
class is
\[
   V_h:=\Sigma_{\rm out}^{J_h},
   \qquad
   |{V_h}|=q_{\rm out}^{sn}.
\]
For each \(c\in C'\), put the hyperedge
\[
   (c[J_1],\ldots,c[J_m])
\]
in \(H\).  This map is injective on \(C'\), because all codewords in \(C'\)
have the same projection on \(J_0\), \(J_0,J_1,\ldots,J_m\) partition
\([n]\), and all codewords in $C'$ are pairwise distinct codewords on $J_0\cup\cdots\cup J_m$.  Hence
\[
   |{E(H)}|=|{C'}|.
\]

Let \(C_*=C_0(m,\ell)\) be the constant supplied by \cref{thm:erdos} such that  if \(H\) contained no
\(K_{\ell,\ldots,\ell}\), then
\[
   |{E(H)}|
   \leq C_*(q_{\rm out}^{sn})^{m-1/\ell^{m-1}}.
\]
But since $R_{\out}-1+\frac{s}{\ell^{m-1}}>0$, there exsits large enough $C_R=C_R(R_{\out},m,\ell)$ so that for all $n>C_R$, we have
\[
\begin{aligned}
   \log_{q_{\rm out}}|{E(H)}|\ge (\rout-1+ms)n> sn(m-\frac{1}{\ell^{m-1}})+\log_{q_{\rm out}}{C_0}.
\end{aligned}
\]
  Therefore, we must have $|{E(H)}|>C_0(q_{\rm out}^{sn})^{m-1/\ell^{m-1}}$
So \(H\) contains a complete \(m\)-partite hypergraph
\(K_{\ell,\ldots,\ell}\).  Thus, for each \(h\in[m]\), there are
\(\ell\) vertices
\[
   p_{h,1},\ldots,p_{h,\ell}\in V_h
\]
such that, for every tuple $d\in[\ell]^m$, the hyperedge
\[
   (p_{1,d[1]},\ldots,p_{m,d[m]})
\]
belongs to \(H\).  Let \(c'_d\in C'\) be the  outer codeword corresponding to this edge.
Then all \(c'_d\)'s agree on \(J_0\) since they are in the same fiber of $g$. Moreover, on \(J_i\colon i\in[m]\), we know that the projected codeword \(c'_d[J_i]\) only depends on $d[i]$.  The hyperedges are
distinct, so the codewords are distinct.
\end{proof}
Remember that our target is to construct $\{c_d\colon d\in[\ell]^m\}\subseteq C_{\rm AEL}$ such that $d_1[i]=d_2[i]\Rightarrow c_{d_1}[I_i]=c_{d_2}[I_i]$ for all $i\in[m], d_1,d_2\in[\ell]^m$. We will choose the encoders $\rm{Enc}_{\rm in,1},\dots,\rm{Enc}_{\rm in,n}\colon \Sigma_{\rm out}\to \Sigma_{\rm in}^D$ of the inner codes $C_{\rm in,1},\dots,C_{\rm in,n}$ and construct $c_d\in C_{\rm AEL}$ as described in the AEL framework: For each edge $e=(u,v)\in E$ which is the $e_u$-th  edge connecting to $u$, let $w_{d,e}=\mathrm{Enc}_{\mathsf{in},u}(c'_{d}[u])[e_u]$. Then, for each right vertex $v\in[n]$, let $e_{v,1},\dots,e_{v,D}\in E$ denote edges connecting to $v$, we define the final codeword $c_d$ as follows
\[
c_d=((w_{d,e_{1,1}},\dots,w_{d,e_{1,D}}),\dots,(w_{d,e_{n,1}},\dots,w_{d,e_{n,D}}))\in (\Sigma_{\mathsf{in}}^D)^n.
\]
Now we state how we choose the inner code encoders. For each $i\in[n]$, there are two cases.
\begin{itemize}
\item If $i\in J_0$, then we choose $\mathsf{Enc}_{\rm in,i}$ to be an arbitrary encoder of $C_{\rm in,i}$.
\item If $i\in J_j$ for some $j\in[m]$, define
\[
N_i\colon=E(\{i\},I\setminus I_j)
\]
Since $i\in J_j\subseteq G_j$, we know that
\[
|N_i|\leq ((m-1)b+\frac{\eps}{2})D=(R_{\rm in}-\frac{\eps}{2})D
\]
Since $D\ge \frac{1}{\lambda^2}\ge\frac{2\log_{2}{\ell}}{\eps}$, let $q_{\mathsf{in}}=|\Sigma_{\mathsf{in}}|$, we know $|C_{\rm in,i}|=q^{R_{\rm in}D}_{\rm in}$ and therefore $q_{\rm in}^{\eps D/2}\ge \ell$. By pigeonhold principle there exist distinct codewords $w_{1,i},\dots,w_{\ell,i}\in C_{\rm in,i}$ that agree on $N'_i=\{e_i\colon e\in N_i\}\subseteq [D]$. We choose an arbitrary encoder $\mathsf{Enc}_{\rm in,i}$ of $C_{\rm in,i}$ such that $\mathsf{Enc}_{\rm in,i}(c'_{f(j,t)}[i])=w_{t,i}$  for all $t\in[\ell]$ where $f(j,t)=1^{j-1}\circ t\circ 1^{m-j}\in[\ell]^m$ denote the unit indicator function.
\end{itemize}
Now we prove that the codewords $\{c_d\colon d\in[\ell]^m\}$ constructed this way satisfy the target: Fix any $i\in[m],d_1,d_2\in[\ell]^m$ such that $d_1[i]=d_2[i]$. We prove $c_{d_1}[v]=c_{d_2}[v]$ for all $v\in I_i$. Fix such an $v$, we need to prove that for any $t\in[D]$
\[
w_{d_1,e_{v,t}}=w_{d_2,e_{v,t}}
\]
Let $e_{v,t}=(u,v)$ and $h=(u,v)_u$, by the construction, it suffices to prove
\[
\mathsf{Enc}_{\rm in,u}(c'_{d_1}[u])[h]=\mathsf{Enc}_{\rm in,u}(c'_{d_2}[u])[h]
\]
There are three cases with respect to $u$
\begin{itemize}
\item $u\in J_0$: Then by our construction $c'_{d_1}[u]=c'_{d_2}[u]$, and the claim immediately holds.
\item $u\in J_i$: Since $d_1[i]=d_2[i]$, we know $c'_{d_1}[J_i]=c'_{d_2}[J_i]$ by \cref{lem:outer-list}, so $c'_{d_1}[u]=c'_{d_2}[u]$ and the claim holds.
\item $u\in J_j$ where $j\in [m]\setminus\{i\}$: Then $e_{v,t}=(u,v)\in N_u$ since $v\in I_i\subseteq I\backslash I_j$ but $u\in J_j$. This means $h\in N'_u$. Since $d_1[j]=f(j,d_1[j])[j]$ and $d_2[j]=f(j,d_2[j])[j]$, we know $c'_{d_1}[u]=c'_{f(j,d_1[j])}[u]$ and $c'_{d_2}[u]=c'_{f(j,d_2[j])}[u]$ by the guarantee of \Cref{lem:outer-list}. Therefore, it suffices to prove that $\mathsf{Enc}_{\rm in,u}(c'_{f(j,d_1[j])}[u])[h]=\mathsf{Enc}_{\rm in,u}(c'_{f(j,d_1[j])}[u])[h]$, which is equivalent to the claim $w_{d_1[j],u}[h]=w_{d_2[j],u}[h]$ by our construction. Since $h\in N'_u$, this equality indeed holds by our choices of $w_{1,u},\dots,w_{\ell,u}$.
\end{itemize}
We need to set $C_0=\max(C_*,C_R,\frac{\log_{\ell}}{R_{\out}-1+s})$ so that the above proof holds for all $n>C_0$.
\end{proof}

\begin{corollary}\label{cor:main}
Let $0<R_{\mathsf{in}},R_{\mathsf{out}}<1,0<\eps<\frac{R_{\mathsf{in}}(1-R_{\mathsf{in}})}{4}$ be any positive constants, $\ell\ge2$ be any positive integer, there exist constants $C_0>0$ and $\lambda>0$ such that for any inner-encoder insensitive AEL framework parameterized by  \[(n,C_{\out},R_{\mathsf{out}},C_{\mathsf{in},1},\dots,C_{\mathsf{in},n},R_{\mathsf{in}},D,\lambda,G)\]  if at least one of the two following conditions holds:
\begin{itemize}
\item $R_{\mathsf{out}}\ge1-\frac{\eps\ell}{R_{\mathsf{in}}\ell^{\lfloor\frac{R_{\mathsf{in}}}{2\eps}\rfloor}}$.
\item $\mathcal{P}_{\mathsf{out}}$ only contains additive codes and $R_{\out}>1-\frac{2\eps}{R_{\mathsf{in}}}$. 
\end{itemize}
then for any large enough $n> C_0$, the prescribed AEL framework does not guarantee to construct a $(1-R-\eps,\ell,\ell^{\lfloor\frac{R}{2\eps}\rfloor}-1)$ list-recoverable code with rate $R=R_{\mathsf{in}}R_{\mathsf{out}}$.
\end{corollary}
\begin{proof}
Let $m=\lfloor\frac{R_{\mathsf{in}}}{2\eps}\rfloor$, by $\eps<\frac{R_{\mathsf{in}}(1-R_{\mathsf{in}})}{4}$ we know $m\ge 2$ and $R_{\mathsf{in}}\leq\frac{m-1}{m}$. Also, there is $1-\frac{m(R_{\mathsf{in}}-\eps)}{m-1}\leq R+\eps$. Therefore, our result follows from  \cref{thm:ael-main}.
\end{proof}

\bibliographystyle{alphaurl}
\bibliography{ref}

\newcommand{\etalchar}[1]{$^{#1}$}
\begin{thebibliography}{KRZSW23}

\bibitem[AEL95]{ael95}
Noga Alon, Jeff Edmonds, and Michael Luby.
\newblock Linear time erasure codes with nearly optimal recovery (extended
  abstract).
\newblock In {\em 36th Annual Symposium on Foundations of Computer Science,
  {FOCS} 1995, Milwaukee, Wisconsin, USA, 23-25 October 1995}, pages 512--519.
  {IEEE} Computer Society, 1995.
\newblock \href {https://doi.org/10.1109/SFCS.1995.492581}
  {\path{doi:10.1109/SFCS.1995.492581}}.

\bibitem[AGG{\etalchar{+}}25]{AGGLZ25}
Omar Alrabiah, Zeyu Guo, Venkatesan Guruswami, Ray Li, and Zihan Zhang.
\newblock Random {R}eed-{S}olomon codes achieve list-decoding capacity with
  linear-sized alphabets.
\newblock {\em Adv. Comb.}, pages Paper No. 8, 39, 2025.

\bibitem[AGL25]{AGL25}
Omar Alrabiah, Venkatesan Guruswami, and Ray Li.
\newblock {AG} codes have no list-decoding friends: Approaching the generalized
  singleton bound requires exponential alphabets.
\newblock {\em {IEEE} Trans. Inf. Theory}, 71(4):2443--2451, 2025.
\newblock \href {https://doi.org/10.1109/TIT.2024.3405392}
  {\path{doi:10.1109/TIT.2024.3405392}}.

\bibitem[AHS26]{ashvinkumar2026algorithmic}
Vikrant Ashvinkumar, Mursalin Habib, and Shashank Srivastava.
\newblock Algorithmic improvements to list decoding of folded reed-solomon
  codes.
\newblock In Kasper~Green Larsen and Barna Saha, editors, {\em Proceedings of
  the 2026 Annual {ACM-SIAM} Symposium on Discrete Algorithms, {SODA} 2026,
  Vancouver, BC, Canada, January 11-14, 2026}, pages 880--898. {SIAM}, 2026.
\newblock \href {https://doi.org/10.1137/1.9781611978971.35}
  {\path{doi:10.1137/1.9781611978971.35}}.

\bibitem[AS08]{AS08}
Noga Alon and Joel~H. Spencer.
\newblock {\em The probabilistic method}.
\newblock Wiley-Interscience Series in Discrete Mathematics and Optimization.
  John Wiley \& Sons, Inc., Hoboken, NJ, third edition, 2008.
\newblock With an appendix on the life and work of Paul Erd\H{o}s.
\newblock \href {https://doi.org/10.1002/9780470277331}
  {\path{doi:10.1002/9780470277331}}.

\bibitem[BCDZ26a]{bcdz26b}
Joshua Brakensiek, Yeyuan Chen, Manik Dhar, and Zihan Zhang.
\newblock Combinatorial bounds for list recovery via discrete brascamp-lieb
  inequalities.
\newblock In Aditya Bhaskara and Artur Czumaj, editors, {\em Proceedings of the
  58th Annual {ACM} Symposium on Theory of Computing, {STOC} 2026, Salt Lake
  City, UT, USA, June 22-26, 2026}, pages 365--376. {ACM}, 2026.
\newblock \href {https://doi.org/10.1145/3798129.3800756}
  {\path{doi:10.1145/3798129.3800756}}.

\bibitem[BCDZ26b]{brakensiek2026random}
Joshua Brakensiek, Yeyuan Chen, Manik Dhar, and Zihan Zhang.
\newblock From random to explicit via subspace designs with applications to
  local properties and matroids.
\newblock In Aditya Bhaskara and Artur Czumaj, editors, {\em Proceedings of the
  58th Annual {ACM} Symposium on Theory of Computing, {STOC} 2026, Salt Lake
  City, UT, USA, June 22-26, 2026}, pages 619--630. {ACM}, 2026.
\newblock \href {https://doi.org/10.1145/3798129.3800778}
  {\path{doi:10.1145/3798129.3800778}}.

\bibitem[BGM23]{bgm23}
Joshua Brakensiek, Sivakanth Gopi, and Visu Makam.
\newblock Generic reed-solomon codes achieve list-decoding capacity.
\newblock In Barna Saha and Rocco~A. Servedio, editors, {\em Proceedings of the
  55th Annual {ACM} Symposium on Theory of Computing, {STOC} 2023, Orlando, FL,
  USA, June 20-23, 2023}, pages 1488--1501. {ACM}, 2023.
\newblock \href {https://doi.org/10.1145/3564246.3585128}
  {\path{doi:10.1145/3564246.3585128}}.

\bibitem[BHV26]{bostanci2026separating}
John Bostanci, Andrew Huang, and Vinod Vaikuntanathan.
\newblock Separating quantum and classical advice with good codes.
\newblock {\em arXiv preprint arXiv:2602.09385}, 2026.

\bibitem[BS14]{note13}
Avraham Ben{-}Aroya and Igor Shinkar.
\newblock A note on subspace evasive sets.
\newblock {\em Chic. J. Theor. Comput. Sci.}, 2014, 2014.
\newblock URL:
  \url{http://cjtcs.cs.uchicago.edu/articles/2014/9/contents.html}.

\bibitem[CZ25]{CZ25}
Yeyuan Chen and Zihan Zhang.
\newblock Explicit folded reed-solomon and multiplicity codes achieve relaxed
  generalized singleton bounds.
\newblock In Michal Kouck{\'{y}} and Nikhil Bansal, editors, {\em Proceedings
  of the 57th Annual {ACM} Symposium on Theory of Computing, {STOC} 2025,
  Prague, Czechia, June 23-27, 2025}, pages 1--12. {ACM}, 2025.
\newblock \href {https://doi.org/10.1145/3717823.3718114}
  {\path{doi:10.1145/3717823.3718114}}.

\bibitem[DL12]{dl12}
Zeev Dvir and Shachar Lovett.
\newblock Subspace evasive sets.
\newblock In Howard~J. Karloff and Toniann Pitassi, editors, {\em Proceedings
  of the 44th Symposium on Theory of Computing Conference, {STOC} 2012, New
  York, NY, USA, May 19 - 22, 2012}, pages 351--358. {ACM}, 2012.
\newblock \href {https://doi.org/10.1145/2213977.2214010}
  {\path{doi:10.1145/2213977.2214010}}.

\bibitem[Eli57]{Elias}
Peter Elias.
\newblock List decoding for noisy channels.
\newblock {\em Wescon Convention Record, Part 2, Institute of Radio Engineers},
  pages 99--104, 1957.

\bibitem[Erd64]{erdos64}
P.~Erd\H{o}s.
\newblock On extremal problems of graphs and generalized graphs.
\newblock {\em Israel J. Math.}, 2:183--190, 1964.
\newblock \href {https://doi.org/10.1007/BF02759942}
  {\path{doi:10.1007/BF02759942}}.

\bibitem[FJ65]{concat}
George~David Forney~Jr.
\newblock {\em Concatenated codes.}
\newblock PhD thesis, Massachusetts Institute of Technology, 1965.

\bibitem[GGH26]{ggh26}
Rohan Goyal, Venkatesan Guruswami, and Jun{-}Ting Hsieh.
\newblock Explicit constant-alphabet subspace design codes.
\newblock {\em CoRR}, abs/2604.15218, 2026.
\newblock URL: \url{https://doi.org/10.48550/arXiv.2604.15218}, \href
  {https://arxiv.org/abs/2604.15218} {\path{arXiv:2604.15218}}, \href
  {https://doi.org/10.48550/ARXIV.2604.15218}
  {\path{doi:10.48550/ARXIV.2604.15218}}.

\bibitem[GI01]{gi01}
Venkatesan Guruswami and Piotr Indyk.
\newblock Expander-based constructions of efficiently decodable codes (extended
  abstract).
\newblock In {\em 42nd {IEEE} {S}ymposium on {F}oundations of {C}omputer
  {S}cience ({L}as {V}egas, {NV}, 2001)}, pages 658--667. IEEE Computer Soc.,
  Los Alamitos, CA, 2001.

\bibitem[GI02]{gi02}
Venkatesan Guruswami and Piotr Indyk.
\newblock Near-optimal linear-time codes for unique decoding and new
  list-decodable codes over smaller alphabets.
\newblock In John~H. Reif, editor, {\em Proceedings on 34th Annual {ACM}
  Symposium on Theory of Computing, May 19-21, 2002, Montr{\'{e}}al,
  Qu{\'{e}}bec, Canada}, pages 812--821. {ACM}, 2002.
\newblock \href {https://doi.org/10.1145/509907.510023}
  {\path{doi:10.1145/509907.510023}}.

\bibitem[GK16]{guruswami2016explicit}
Venkatesan Guruswami and Swastik Kopparty.
\newblock Explicit subspace designs.
\newblock {\em Comb.}, 36(2):161--185, 2016.
\newblock URL: \url{https://doi.org/10.1007/s00493-014-3169-1}, \href
  {https://doi.org/10.1007/S00493-014-3169-1}
  {\path{doi:10.1007/S00493-014-3169-1}}.

\bibitem[GKdO{\etalchar{+}}18]{gko}
Sivakanth Gopi, Swastik Kopparty, Rafael~Mendes de~Oliveira, Noga Ron{-}Zewi,
  and Shubhangi Saraf.
\newblock Locally testable and locally correctable codes approaching the
  gilbert-varshamov bound.
\newblock {\em {IEEE} Trans. Inf. Theory}, 64(8):5813--5831, 2018.
\newblock \href {https://doi.org/10.1109/TIT.2018.2809788}
  {\path{doi:10.1109/TIT.2018.2809788}}.

\bibitem[GNP{\etalchar{+}}13]{GilbertNgoPoratStrauss13}
Anna~C. Gilbert, Hung~Q. Ngo, Ely Porat, Atri Rudra, and Martin~J. Strauss.
\newblock {$\ell_2/\ell_2$}-foreach sparse recovery with low risk.
\newblock In {\em Automata, Languages, and Programming -- 40th International
  Colloquium, ICALP 2013, Proceedings, Part I}, volume 7965 of {\em Lecture
  Notes in Computer Science}, pages 461--472. Springer, 2013.
\newblock \href {https://doi.org/10.1007/978-3-642-39206-1_39}
  {\path{doi:10.1007/978-3-642-39206-1_39}}.

\bibitem[GR08]{gr08}
Venkatesan Guruswami and Atri Rudra.
\newblock Explicit codes achieving list decoding capacity: Error-correction
  with optimal redundancy.
\newblock {\em {IEEE} Trans. Inf. Theory}, 54(1):135--150, 2008.
\newblock \href {https://doi.org/10.1109/TIT.2007.911222}
  {\path{doi:10.1109/TIT.2007.911222}}.

\bibitem[GS98]{guruswami1998improved}
Venkatesan Guruswami and Madhu Sudan.
\newblock Improved decoding of {R}eed-{S}olomon and algebraic-geometric codes.
\newblock In {\em 39th Annual Symposium on Foundations of Computer Science,
  {FOCS} 1998, Palo Alto, California, USA, November 8-11, 1998}, pages 28--39.
  {IEEE} Computer Society, 1998.
\newblock \href {https://doi.org/10.1109/SFCS.1998.743426}
  {\path{doi:10.1109/SFCS.1998.743426}}.

\bibitem[GST24]{GST24}
Eitan Goldberg, Chong Shangguan, and Itzhak Tamo.
\newblock Singleton-type bounds for list-decoding and list-recovery, and
  related results.
\newblock {\em J. Combin. Theory Ser. A}, 203:Paper No. 105835, 26, 2024.
\newblock \href {https://doi.org/10.1016/j.jcta.2023.105835}
  {\path{doi:10.1016/j.jcta.2023.105835}}.

\bibitem[Gur11]{G11}
Venkatesan Guruswami.
\newblock Linear-algebraic list decoding of folded reed-solomon codes.
\newblock In {\em Proceedings of the 26th Annual {IEEE} Conference on
  Computational Complexity, {CCC} 2011, San Jose, California, USA, June 8-10,
  2011}, pages 77--85. {IEEE} Computer Society, 2011.
\newblock \href {https://doi.org/10.1109/CCC.2011.22}
  {\path{doi:10.1109/CCC.2011.22}}.

\bibitem[GUV09]{guruswami2009unbalanced}
Venkatesan Guruswami, Christopher Umans, and Salil~P. Vadhan.
\newblock Unbalanced expanders and randomness extractors from parvaresh-vardy
  codes.
\newblock {\em J. {ACM}}, 56(4):20:1--20:34, 2009.
\newblock \href {https://doi.org/10.1145/1538902.1538904}
  {\path{doi:10.1145/1538902.1538904}}.

\bibitem[GW11]{GW11}
Venkatesan Guruswami and Carol Wang.
\newblock Optimal rate list decoding via derivative codes.
\newblock In Leslie~Ann Goldberg, Klaus Jansen, R.~Ravi, and Jos{\'{e}} D.~P.
  Rolim, editors, {\em Approximation, Randomization, and Combinatorial
  Optimization. Algorithms and Techniques - 14th International Workshop,
  {APPROX} 2011, and 15th International Workshop, {RANDOM} 2011, Princeton, NJ,
  USA, August 17-19, 2011. Proceedings}, Lecture Notes in Computer Science,
  pages 593--604. Springer, 2011.
\newblock \href {https://doi.org/10.1007/978-3-642-22935-0\_50}
  {\path{doi:10.1007/978-3-642-22935-0\_50}}.

\bibitem[HIOS15]{haitner2015parallel}
Iftach Haitner, Yuval Ishai, Eran Omri, and Ronen Shaltiel.
\newblock Parallel hashing via list recoverability.
\newblock In Rosario Gennaro and Matthew Robshaw, editors, {\em Advances in
  Cryptology - {CRYPTO} 2015 - 35th Annual Cryptology Conference, Santa
  Barbara, CA, USA, August 16-20, 2015, Proceedings, Part {II}}, volume 9216 of
  {\em Lecture Notes in Computer Science}, pages 173--190. Springer, 2015.
\newblock \href {https://doi.org/10.1007/978-3-662-48000-7\_9}
  {\path{doi:10.1007/978-3-662-48000-7\_9}}.

\bibitem[HRZW19]{hrw}
Brett Hemenway, Noga Ron-Zewi, and Mary Wootters.
\newblock Local list recovery of high-rate tensor codes and applications.
\newblock {\em SIAM Journal on Computing}, 49(4):FOCS17--157, 2019.

\bibitem[HW18]{hw15}
Brett Hemenway and Mary Wootters.
\newblock Linear-time list recovery of high-rate expander codes.
\newblock {\em Inform. and Comput.}, 261(part 2):202--218, 2018.
\newblock \href {https://doi.org/10.1016/j.ic.2018.02.004}
  {\path{doi:10.1016/j.ic.2018.02.004}}.

\bibitem[JMST25]{jmst25}
Fernando~Granha Jeronimo, Tushant Mittal, Shashank Srivastava, and Madhur
  Tulsiani.
\newblock Explicit codes approaching generalized singleton bound using
  expanders.
\newblock In {\em S{TOC}'25---{P}roceedings of the 57th {A}nnual {ACM}
  {S}ymposium on {T}heory of {C}omputing}, pages 843--854. ACM, New York,
  [2025] \copyright 2025.
\newblock \href {https://doi.org/10.1145/3717823.3718302}
  {\path{doi:10.1145/3717823.3718302}}.

\bibitem[JS26]{js26}
Fernando~Granha Jeronimo and Nikhil Shagrithaya.
\newblock Probabilistic guarantees to explicit constructions: Local properties
  of linear codes.
\newblock In Aditya Bhaskara and Artur Czumaj, editors, {\em Proceedings of the
  58th Annual {ACM} Symposium on Theory of Computing, {STOC} 2026, Salt Lake
  City, UT, USA, June 22-26, 2026}, pages 806--813. {ACM}, 2026.
\newblock \href {https://doi.org/10.1145/3798129.3800795}
  {\path{doi:10.1145/3798129.3800795}}.

\bibitem[KMRS17]{highrate}
Swastik Kopparty, Or~Meir, Noga Ron{-}Zewi, and Shubhangi Saraf.
\newblock High-rate locally correctable and locally testable codes with
  sub-polynomial query complexity.
\newblock {\em J. {ACM}}, 64(2):11:1--11:42, 2017.
\newblock \href {https://doi.org/10.1145/3051093} {\path{doi:10.1145/3051093}}.

\bibitem[Kop15]{kopparty2015list}
Swastik Kopparty.
\newblock List-decoding multiplicity codes.
\newblock {\em Theory Comput.}, 11:149--182, 2015.
\newblock URL: \url{https://doi.org/10.4086/toc.2015.v011a005}, \href
  {https://doi.org/10.4086/TOC.2015.V011A005}
  {\path{doi:10.4086/TOC.2015.V011A005}}.

\bibitem[KRZSW23]{krsw}
Swastik Kopparty, Noga Ron-Zewi, Shubhangi Saraf, and Mary Wootters.
\newblock Improved list decoding of folded {R}eed-{S}olomon and multiplicity
  codes.
\newblock {\em SIAM J. Comput.}, 52(3):794--840, 2023.
\newblock \href {https://doi.org/10.1137/20M1370215}
  {\path{doi:10.1137/20M1370215}}.

\bibitem[LMS25]{LMS25}
Matan Levi, Jonathan Mosheiff, and Nikhil Shagrithaya.
\newblock Random reed-solomon codes and random linear codes are locally
  equivalent.
\newblock In {\em 66th {IEEE} Annual Symposium on Foundations of Computer
  Science, {FOCS} 2025, Sydney, Australia, December 14-17, 2025}, pages
  2097--2131. {IEEE}, 2025.
\newblock \href {https://doi.org/10.1109/FOCS63196.2025.00112}
  {\path{doi:10.1109/FOCS63196.2025.00112}}.

\bibitem[LPS88]{lps88}
Alexander Lubotzky, Ralph Phillips, and Peter Sarnak.
\newblock Ramanujan graphs.
\newblock {\em Combinatorica}, 8(3):261--277, 1988.

\bibitem[LS25]{LS25}
Ray Li and Nikhil Shagrithaya.
\newblock Near-optimal list-recovery of linear code families.
\newblock In Alina Ene and Eshan Chattopadhyay, editors, {\em Approximation,
  Randomization, and Combinatorial Optimization. Algorithms and Techniques,
  {APPROX/RANDOM} 2025, Berkeley, CA, USA, August 11-13, 2025}, LIPIcs, pages
  53:1--53:14. Schloss Dagstuhl - Leibniz-Zentrum f{\"{u}}r Informatik, 2025.
\newblock \href {https://doi.org/10.4230/LIPIcs.APPROX/RANDOM.2025.53}
  {\path{doi:10.4230/LIPIcs.APPROX/RANDOM.2025.53}}.

\bibitem[NPR12]{ngo2012efficiently}
Hung~Q. Ngo, Ely Porat, and Atri Rudra.
\newblock Efficiently decodable compressed sensing by list-recoverable codes
  and recursion.
\newblock In Christoph D{\"{u}}rr and Thomas Wilke, editors, {\em 29th
  International Symposium on Theoretical Aspects of Computer Science, {STACS}
  2012, Paris, France, February 29 - March 3, 2012}, volume~14 of {\em LIPIcs},
  pages 230--241. Schloss Dagstuhl - Leibniz-Zentrum f{\"{u}}r Informatik,
  2012.
\newblock URL: \url{https://doi.org/10.4230/LIPIcs.STACS.2012.230}, \href
  {https://doi.org/10.4230/LIPICS.STACS.2012.230}
  {\path{doi:10.4230/LIPICS.STACS.2012.230}}.

\bibitem[PR04]{pr04}
Pavel Pudl\'{a}k and Vojt\v{e}ch R\"{o}dl.
\newblock Pseudorandom sets and explicit constructions of {R}amsey graphs.
\newblock In {\em Complexity of computations and proofs}, volume~13 of {\em
  Quad. Mat.}, pages 327--346. Dept. Math., Seconda Univ. Napoli, Caserta,
  2004.

\bibitem[RR97]{RR97}
Alexander~A. Razborov and Steven Rudich.
\newblock Natural proofs.
\newblock {\em J. Comput. Syst. Sci.}, 55(1):24--35, 1997.
\newblock URL: \url{https://doi.org/10.1006/jcss.1997.1494}, \href
  {https://doi.org/10.1006/JCSS.1997.1494} {\path{doi:10.1006/JCSS.1997.1494}}.

\bibitem[RV25]{rv26}
Nicolas Resch and S.~Venkitesh.
\newblock List recoverable codes: The good, the bad, and the unknown (hopefully
  not ugly).
\newblock {\em CoRR}, abs/2510.07597, 2025.
\newblock URL: \url{https://doi.org/10.48550/arXiv.2510.07597}, \href
  {https://arxiv.org/abs/2510.07597} {\path{arXiv:2510.07597}}, \href
  {https://doi.org/10.48550/ARXIV.2510.07597}
  {\path{doi:10.48550/ARXIV.2510.07597}}.

\bibitem[RYZ24]{RYZ24}
Nicolas Resch, Chen Yuan, and Yihan Zhang.
\newblock Zero-rate thresholds and new capacity bounds for list-decoding and
  list-recovery.
\newblock {\em {IEEE} Trans. Inf. Theory}, 70(9):6211--6238, 2024.
\newblock \href {https://doi.org/10.1109/TIT.2024.3430842}
  {\path{doi:10.1109/TIT.2024.3430842}}.

\bibitem[Sin64]{singleton}
Richard~C. Singleton.
\newblock Maximum distance q-nary codes.
\newblock {\em {IEEE} Trans. Inf. Theory}, 10(2):116--118, 1964.
\newblock \href {https://doi.org/10.1109/TIT.1964.1053661}
  {\path{doi:10.1109/TIT.1964.1053661}}.

\bibitem[Sri25]{srivastava2025improved}
Shashank Srivastava.
\newblock Improved list size for folded reed-solomon codes.
\newblock In Yossi Azar and Debmalya Panigrahi, editors, {\em Proceedings of
  the 2025 Annual {ACM-SIAM} Symposium on Discrete Algorithms, {SODA} 2025, New
  Orleans, LA, USA, January 12-15, 2025}, pages 2040--2050. {SIAM}, 2025.
\newblock \href {https://doi.org/10.1137/1.9781611978322.64}
  {\path{doi:10.1137/1.9781611978322.64}}.

\bibitem[ST20]{st20}
Chong Shangguan and Itzhak Tamo.
\newblock Combinatorial list-decoding of reed-solomon codes beyond the johnson
  radius.
\newblock In Konstantin Makarychev, Yury Makarychev, Madhur Tulsiani, Gautam
  Kamath, and Julia Chuzhoy, editors, {\em Proceedings of the 52nd Annual {ACM}
  {SIGACT} Symposium on Theory of Computing, {STOC} 2020, Chicago, IL, USA,
  June 22-26, 2020}, pages 538--551. {ACM}, 2020.
\newblock \href {https://doi.org/10.1145/3357713.3384295}
  {\path{doi:10.1145/3357713.3384295}}.

\bibitem[ST25]{st25}
Shashank Srivastava and Madhur Tulsiani.
\newblock List decoding expander-based codes up to capacity in near-linear
  time.
\newblock In {\em 66th {IEEE} Annual Symposium on Foundations of Computer
  Science, {FOCS} 2025, Sydney, Australia, December 14-17, 2025}, pages
  1465--1487. {IEEE}, 2025.
\newblock \href {https://doi.org/10.1109/FOCS63196.2025.00077}
  {\path{doi:10.1109/FOCS63196.2025.00077}}.

\bibitem[Sud97]{sudan1997decoding}
Madhu Sudan.
\newblock Decoding of {R}eed {S}olomon codes beyond the error-correction bound.
\newblock {\em J. Complex.}, 13(1):180--193, 1997.
\newblock URL: \url{https://doi.org/10.1006/jcom.1997.0439}, \href
  {https://doi.org/10.1006/JCOM.1997.0439} {\path{doi:10.1006/JCOM.1997.0439}}.

\bibitem[Tam24]{tamo24}
Itzhak Tamo.
\newblock Tighter list-size bounds for list-decoding and recovery of folded
  reed-solomon and multiplicity codes.
\newblock {\em {IEEE} Trans. Inf. Theory}, 70(12):8659--8668, 2024.
\newblock \href {https://doi.org/10.1109/TIT.2024.3402171}
  {\path{doi:10.1109/TIT.2024.3402171}}.

\bibitem[TUZ07]{ta2007lossless}
Amnon Ta{-}Shma, Christopher Umans, and David Zuckerman.
\newblock Lossless condensers, unbalanced expanders, and extractors.
\newblock {\em Comb.}, 27(2):213--240, 2007.
\newblock URL: \url{https://doi.org/10.1007/s00493-007-0053-2}, \href
  {https://doi.org/10.1007/S00493-007-0053-2}
  {\path{doi:10.1007/S00493-007-0053-2}}.

\bibitem[TVZ82]{tvz82}
M.~A. Tsfasman, S.~G. Vl\u{a}du\c{t}, and Th. Zink.
\newblock Modular curves, {S}himura curves, and {G}oppa codes, better than
  {V}arshamov-{G}ilbert bound.
\newblock {\em Math. Nachr.}, 109:21--28, 1982.
\newblock \href {https://doi.org/10.1002/mana.19821090103}
  {\path{doi:10.1002/mana.19821090103}}.

\bibitem[Woz58]{wozencraft1958list}
John~M Wozencraft.
\newblock List decoding.
\newblock {\em Quarterly Progress Report}, 48:90--95, 1958.

\end{thebibliography}

\end{document}